\documentclass[11pt]{amsart}

\usepackage{geometry}

\usepackage{amssymb}
\usepackage{amsmath}
\usepackage{centernot}
\usepackage{graphicx} 
\usepackage{enumerate}
\usepackage{xspace}
\usepackage{enumitem}
\usepackage{amsthm}
\usepackage{lmodern}
\usepackage{dsfont}
\usepackage{mathrsfs}
\usepackage[scr=boondoxo]{mathalfa}
\usepackage{bbm}
\usepackage{xcolor}    % For colored links

\usepackage{hyperref}
\hypersetup{
	colorlinks=true,
	allcolors=blue, 
}
\usepackage[numbers]{natbib}
\usepackage{microtype}

\usepackage{mathtools}
\mathtoolsset{showonlyrefs}

\makeatletter
\let\old@abstract\abstract
\let\old@endabstract\endabstract
\renewenvironment{abstract}{%
	\old@abstract
	\small
}{%
	\old@endabstract
}
\makeatother

\theoremstyle{plain}
\newtheorem{theorem}{Theorem}[section]
\newtheorem{corollary}[theorem]{Corollary}
\newtheorem{lemma}[theorem]{Lemma}

\newtheorem{proposition}[theorem]{Proposition}

\theoremstyle{definition}
\newtheorem{definition}[theorem]{Definition}

\newtheorem{SA}[theorem]{Standing Assumption}

\newtheorem{example}[theorem]{Example}

\theoremstyle{remark}
\newtheorem{remark}[theorem]{Remark}
\newtheorem{discussion}[theorem]{Discussion}

\numberwithin{equation}{section}

\DeclareMathOperator{\IR}{\mathbb{R}} %{st.}
\newcommand{\rd}{\mathrm{d}}
\newcommand{\vd}{\,\mathrm{d}}

\newcommand{\on}{\operatorname}
\newcommand{\m}{\mathfrak{m}}
\newcommand{\s}{\mathfrak{s}}
\newcommand{\bR}{\mathbb{R}}
\newcommand{\q}{\mathfrak{q}}
\newcommand{\si}{{\on{si}}}
\newcommand{\ac}{{\on{ac}}}
\newcommand{\1}{\mathbbm{1}} %{\mathbb{I}}
\newcommand{\cF}{\mathcal{F}}
\renewcommand{\P}{\mathbb{P}} %{\mathsf{P}}
\newcommand{\Y}{Y}
\newcommand{\U}{U}
\renewcommand{\S}{S}

\renewcommand{\l}{\alpha}
\renewcommand{\r}{\beta}

\newcommand{\otheta}{\bar{\theta}}

\newcommand{\fmu}{\m^U_\ac} %{\widetilde\m^\U}
\newcommand{\fqpp}{\q''_\ac} %{\tilde\q''}

\newcommand{\IP}{\mathsf{IP}}
\newcommand{\QVIP}{\mathsf{QVIP}}
\newcommand{\cA}{\mathcal{A}}

\newcommand\llambda{{\mathchoice
		{\lambda\mkern-4.5mu{\raisebox{.4ex}{\scriptsize$\backslash$}}}
		{\lambda\mkern-4.83mu{\raisebox{.4ex}{\scriptsize$\backslash$}}}
		{\lambda\mkern-4.5mu{\raisebox{.2ex}{\footnotesize$\scriptscriptstyle\backslash$}}}
		{\lambda\mkern-5.0mu{\raisebox{.2ex}{\tiny$\scriptscriptstyle\backslash$}}}}}

\newcommand{\indic}[1]{{\mathbbm{1}}_{#1}}

\newenvironment{enuroman}
{\begin{enumerate}[label=\textup{(\roman*)}]}
	{\end{enumerate}}
\newcommand{\hemail}[1]{\href{mailto:{#1}}{#1}}

\begin{document}
	
	\title[On the structure of increasing profits in a general diffusion market]{On the structure of increasing profits in a 1D general diffusion market with interest rates}
	
	\author[A. Anagnostakis]{Alexis Anagnostakis}
	\address{
		A. Anagnostakis -- Université de Lorraine, CNRS, IECL 
		F-57000, Metz, France.
	}
	\email{\hemail{alexis.anagnostakis@univ-grenoble-alpes.fr}}
	
	\author[D. Criens]{David Criens}
	\address{D. Criens -- University of Freiburg, Ernst-Zermelo-Str. 1, 79104 Freiburg, Germany.}
	\email{\hemail{david.criens@stochastik.uni-freiburg.de}}
	
	\author[M. Urusov]{Mikhail Urusov}
	\address{M. Urusov -- University of Duisburg-Essen, Thea-Leymann-Str. 9, 45127 Essen, Germany.}
	\email{\hemail{mikhail.urusov@uni-due.de}}
	
	
	\makeatletter
	\@namedef{subjclassname@2020}{\textup{2020} Mathematics Subject Classification}
	\makeatother
	
	
	\thanks{}
	
	\date{\today}
	
	\allowdisplaybreaks
	\frenchspacing
	
	\begin{abstract}
In this paper, we investigate a financial market model consisting of a risky asset, modeled as a general diffusion parameterized by a scale function and a speed measure, and a bank account process with a constant interest rate. This flexible class of financial market models allows for features such as reflecting boundaries, skewness effects, sticky points, and slowdowns on fractal sets. For this market model, we study the structure of a strong form of arbitrage opportunity called {\em increasing profits}. Our main contributions are threefold. First, we characterize the existence of increasing profits in terms of an auxiliary deterministic signed measure $\nu$ and a canonical trading strategy $\theta$, both of which depend only on the deterministic parametric characteristics of our model, namely the scale function, the speed measure, and the interest rate. More precisely, we show that an increasing profit exists if and only if $\nu$ is nontrivial, and that this is equivalent to $\theta$ itself generating an increasing profit. Second, we provide a precise characterization of the entire set of increasing profits in terms of $\nu$ and $\theta$, and moreover characterize the value processes associated with increasing profits. Finally, we establish novel connections between no-arbitrage theory and the general theory of stochastic processes. Specifically, we relate the failure of the representation property for general diffusions to the existence of certain types of increasing profits whose value processes are dominated by the quadratic variation measure of a space-transformed version of the asset price process.
	\end{abstract}
	
	\maketitle

	%%%%%% SECTION 1
	\section{Introduction}
	\label{sec_intro}
	%%%%%%%%%%%%%%%%
	
	Diffusion models with non-standard path properties, such as reflection, stickiness or skewness, have proven to be valuable tools for modeling a wide range of economic and financial scenarios. Prominent examples for applications of models with reflecting boundaries include portfolio protection mechanisms, where capital is added once the portfolio value reaches a prescribed threshold (\cite{Gerber_00}); withdrawal strategies designed to secure a minimal level of income prior to retirement (\cite{KSW_10}); or situations in which central bank interventions aim to maintain exchange rates above a lower bound (\cite{NS_16}). Furthermore, diffusions with sticky points are able to capture possible takeover offers (\cite{CU22}) and models with skewness naturally arise in the context of local volatility models and have been linked to the so-called “steep short end of the smile” phenomenon~(\cite{GS_17,P_19}).
	
	The recent paper \cite{BDH23} has drawn attention to the fact that such models may admit particularly strong forms of arbitrage, so-called {\em increasing profits}, characterized by an increasing value process whose terminal value is positive with positive probability. In our previous article \cite{ACU25}, we provided a comprehensive analysis of the existence and absence of increasing profits for general one-dimensional diffusion models with a single risky asset, modeled as a general diffusion in the sense of It\^o and McKean \cite{ItoMcKean96}, and a bank account with constant interest rate. Our results provide a characterization of the \emph{no increasing profit (NIP)} condition in terms of the deterministic characteristics of the underlying general diffusion, the scale function and the speed measure. 
	
	In the present paper, we go a step further and study the {\em structural foundations} of increasing profits within such a general diffusion framework. Our goal is a precise description of the set of increasing profits and their corresponding value processes. In this regard, our results are particularly relevant in models where the no-arbitrage condition NIP fails, such as the diffusion models studied in \cite{GS_17,Gerber_00,KSW_10,NS_16,P_19}, where increasing profits naturally occur.
	Our analysis also provides an alternative approach to our previous deterministic characterization of NIP from \cite{ACU25} whose proofs relied on the fundamental weak structure conditions (\cite{Fontana15}), which are not used in the present paper. 
	
	Let us explain the results from this paper in a more precise manner.
	Our contributions are threefold. The main mathematical objects of interest are an auxiliary signed measure \(\nu\) and a related trading strategy \(\theta\), both depending only on the interest rate, the scale function and speed measure of the risky asset. 
	
	Our first main result reveals their canonical importance showing that an increasing profit exists if and only if \(\theta\) itself constitutes an increasing profit. We further show that this is equivalent to the fact that \(\nu\) is a nontrivial signed measure.
	The second part of this result recovers our characterization for NIP that we established in \cite{ACU25}. 
	By means of several examples, including Black--Scholes and Bachelier models featuring absorbing or reflecting boundaries, sticky points and skewness effects, we relate each component of \(\nu\) to specific path properties of the underlying diffusion, thereby clarifying exactly which path properties lead to increasing profits and how they can be exploited. In this context, we make the interesting observation that although value processes of increasing profits cannot be dominated by the quadratic variation measure of the asset price process (as otherwise they would be zero identically; see Lemma~\ref{lem: no domination} below), it is possible that such value processes are dominated by the quadratic variation measure of a space transformation of the asset process, a diffusion on natural scale. We call such strategies quadratic variation increasing profits. This appears to be a curiosity of our general diffusion framework, as this feature cannot be observed in classical SDE models, even under the very weak Engelbert--Schmidt conditions.  
		
		As a second main contribution, we characterize the entire class of increasing profits in relation to the canonical strategy \(\theta\) and the signed measure \(\nu\). In particular, we show that increasing profits can only be generated during times when \(\theta\) does not vanish, which underlines its structural importance.
		Moreover, we obtain an explicit representation for the value processes of any increasing profit, linking them to \(\nu\) and to the local time process of the underlying diffusion.
		
		Lastly, we reveal an intrinsic relation between the failure of the representation property (RP)
		for the risky asset and the existence of quadratic variation increasing profits.
        The RP is known to be of fundamental importance in the context of market completeness and also from the viewpoint of the general theory of stochastic processes.
		More specifically, we identify a broad framework in which the existence of quadratic variation increasing profits is equivalent to the failure of the~RP. In general, however, the NIP condition and the RP are in a general position, meaning that neither of them implies the other.

	\emph{Outline.}\quad The paper is structured as follows. In Section~\ref{sec_preliminaries} we introduce our financial market model and recall the notion of increasing profits.
	 Our main results are presented in Section~\ref{sec: main}, the proofs are deferred to Section~\ref{sec:proofs}, and the examples are discussed in Section~\ref{sec: examples}. 
	In the concluding Section~\ref{sec: RP} we discuss the relation between the existence of quadratic variation increasing profits and the failure of the representation property.

	%%%%%% SECTION 2
	\section{The Financial Market and the Concept of Increasing Profits}
	\label{sec_preliminaries}
	%%%%%%%%%%%%%%%%

	\subsection{The Financial Market}
	In this paper, we consider a financial market driven by a regular
	continuous strong Markov process, which is alternatively called a
	\emph{general diffusion}.
	A quite complete overview on the theory of general diffusions can be found in the seminal monograph \cite{ItoMcKean96} by It\^o and McKean.
	Shorter textbook introductions are given in \cite{freedman,Kal21,RevYor,RW2}.
	
	As the concepts of scale and speed are crucial for our results, we recall some facts about them without going too much into detail.
	We take a state space \(J \subset \bR\) that is supposed to be a bounded or unbounded, closed, open or half-open interval. A scale function is a strictly increasing continuous function $\s\colon J\to\mathbb R$
	and a speed measure is a measure $\m$ on $(J,\mathcal B(J))$ that satisfies
	$\m([a,b])\in(0,\infty)$ for all $a<b$ in $J^\circ$, where \(J^\circ\) denotes the interior of~\(J\). We define
	$$
	\l\triangleq\inf J\in[-\infty,\infty)
	\quad\text{and}\quad
	\r\triangleq\sup J\in(-\infty,\infty].
	$$
	The values $\s(\l)$ and $\s(\r)$ are defined by continuity (in particular, they can be infinite).
	We also remark that the speed measure can be infinite near $\l$ and $\r$, and that the values $\m(\{\l\})$ and \(\m (\{\r\})\) can be anything in $[0,\infty]$ provided $\l\in J$ and $\r\in J$, respectively. 
	
	Before we proceed, let us mention that speed measures and semimartingale local times are not scaled consistently in the literature.
	For the speed measure, we use the scaling from the books of Kallenberg~\cite{Kal21} and Rogers and Williams~\cite{RW2}, which is half the speed measure from the monographs of 
	Freedman~\cite{freedman},
	It\^o and McKean~\cite{ItoMcKean96}
	and Revuz and Yor \cite{RevYor}. To give an example, our speed measure of Brownian motion (on natural scale) is simply the Lebesgue measure, while it is twice the Lebesgue measure in \cite{freedman,ItoMcKean96,RevYor}. Similarly, we use the semimartingale local time scaling of Freedman~\cite{freedman}, Kallenberg~\cite{Kal21}, Revuz and Yor~\cite{RevYor} and Rogers and Williams~\cite{RW2}, which is twice the local time of It\^o and McKean~\cite{ItoMcKean96} and Karatzas and Shreve~\cite{KarShr}. 
	Furthermore, we emphasize that we always use the right-continuous version of the semimartingale local time
	(in the space variable).
	
	\smallskip 
	We are in a position to explain our financial framework.
	Throughout this paper, we consider a finite time horizon \(T \in (0, \infty)\).
	Let \(\mathbb{B} = (\Omega, \cF, \mathbf{F} = (\cF_t)_{t \in [0, T]}, \P)\) be a filtered probability space with a right-continuous filtration that supports a regular continuous strong Markov process (in the sense of \cite[Section~V.45]{RW2} except that the underlying setting needs not to be the canonical one) \(\Y = (\Y_t)_{t \in [0, T]}\) with state space~\(J\), scale function~\(\s\), speed measure~\(\m\) and deterministic starting value
	$x_0$. As for the starting value, we always assume that
	$$
	\text{either }x_0\in J^\circ\text{ or }x_0\in J\setminus J^\circ\text{ is a reflecting boundary for }Y.
	$$
	We exclude the case of an absorbing starting value $x_0\in J\setminus J^\circ$, since then the process \(\Y\) is simply constant.
	In the above context, the strong Markov property refers to the filtration~\(\mathbf{F}\).
	
	\begin{SA} \label{SA: semi + boundary}
		\(\Y\) is a semimartingale on the stochastic basis \(\mathbb{B}\).
	\end{SA}
	
	The Standing Assumption~\ref{SA: semi + boundary} is not automatically true in our general diffusion setting. For example, if \(B\) is a Brownian motion starting in zero, then $\sqrt{|B|}$ is a general diffusion but \emph{not} a semimartingale (\cite[Exercise~VI.1.14]{RevYor}). 
	The semimartingale property of \(\Y\) is solely a property of the scale function \(\s\), more precisely, but equivalently, its inverse.
	The following lemma collects some properties that are proved in \cite[Section~5]{CinJacProSha80}.

	Recall that for an open interval $I\subset\bR$ and a real-valued function \(\mathfrak{f}\colon I\to\bR\) that is the difference of two convex functions on $I$, one can define the second derivative measure \(\mathfrak{f}'' (\rd x)\) by 
	\[
	\mathfrak{f}'' ((x, y]) \triangleq \mathfrak{f}'_+ (y) - \mathfrak{f}'_+(x), \quad x < y\text{ in }I, 
	\]
	where \(\mathfrak{f}'_+\) denotes the right derivative of~\(\mathfrak{f}\). 
	
	\begin{lemma} \label{lem: isf DC}
		Assume that \(\Y\) is a semimartingale. Then, the inverse scale function \(\q \triangleq \s^{-1}\) is the difference of two convex functions on the interior \(\s (J^\circ)\). Furthermore, in case \(J = [\l, \infty)\) and \(\l\) is absorbing for \(\Y\), it holds that
		\[
		\int_{\s (\l) +} (x - \s (\l)) \, |\q''| (\rd x) < \infty. 
		\]
		In case \(J = [\l, \infty)\) and \(\l\) is reflecting for \(\Y\), the second derivative measure \(\q'' (\rd x)\) can be identified with a finite signed measure on every interval \([\s (\l), \s (z)]\) with \(z \in (\l,\infty)\). 
	\end{lemma}
	
	\begin{proof}
		These statements follow directly from the discussion in \cite[Section~5]{CinJacProSha80}. 
	\end{proof}
	Of course, suitable adjustments of the last two statements from Lemma~\ref{lem: isf DC} hold also for more general state spaces~\(J\).

	\smallskip
	In the following, our financial market is supposed to contain one risky asset that is given by the general diffusion semimartingale \(\Y\). 
	Furthermore, we fix a deterministic interest rate \(r \in \bR\). The discounting will be done by the usual bank account process \(e^{r t}\) for \(t \in [0, T]\), leading to the discounted price process \(\S = (\S_t)_{t \in [0, T]}\) given by
	\begin{align} \label{eq:311025a1}
		\S_t \triangleq e^{- rt} \, \Y_t, \quad t \in [0, T].
	\end{align}
	We now recall the concept of increasing profits, which is the central arbitrage concept studied in this paper.

	\subsection{Increasing Profits} 
	Let us now recall the ``no increasing profit'' NIP condition, which is similar to the ``no unbounded increasing profit'' first introduced by Karatzas and Kardaras in \cite{KaraKard07}.
	Our presentation follows Fontana \cite{Fontana15}.
	In the sequel we use the notation $L(\S)$ for the set of all predictable processes that are integrable w.r.t. the continuous semimartingale \(\S\).
	The elements $H\in L(\S)$ are alternatively called \emph{strategies}. The integral process
	\[
	V^H_t \triangleq \int_0^t H_s \vd \S_s, \quad t \in [0, T], 
	\] 
	is called the {\em value process} associated to the strategy \(H \in L (\S)\).

	\begin{definition}
		A strategy \(H \in L(\S)\) is called an {\em increasing profit} if 
		\begin{equation}
			\begin{aligned}
				\P\text{-a.s.},\; [0, T] \ni t \mapsto V^H_t \text{ is increasing, and } \P\big( V^H_T > 0 \big) > 0.
			\end{aligned}
		\end{equation}
		We denote the set of all such strategies by $\IP $:
		\[
		\IP \triangleq \big\{ H \in L (\S) \colon H \text{ is an increasing profit} \big\}.
		\] 
		If \(\IP = \emptyset\), we say that the \emph{NIP} condition holds. 
	\end{definition}
	
	\begin{remark} \label{rem: unbounded} 
		If \(\IP\neq \emptyset\), the set \(\{ V^H_T \colon H \in \IP \}\) is unbounded with positive probability, i.e., 
		\[
		\P \Big( \sup_{H \, \in \, \IP} V^H_T = \infty \Big) > 0.
		\]  
		This follows directly from the fact that \(\IP\) is a cone.
	\end{remark}
		
		\begin{remark} \label{rem: FV profit}
			While increasing profits have an obvious financial interpretation,
			from the viewpoint of the general theory of stochastic processes
			the following question seems to be more natural:
			\emph{When does a strategy $H\in L(\S)$ exist such that $V^H$ is a non-constant finite variation process?}
			
			We recall the answer to this question in the more general setting where the discounted price is a continuous semimartingale
			(notice that, from the next section on, we again consider the discounted general diffusion setting~\eqref{eq:311025a1} only):
			
			Let $S'=(S'_t)_{t\in[0,T]}$ be a continuous semimartingale.
			Then there exists an increasing profit $H\in L(\S')$
			if and only if
			there exists a trading strategy $K\in L(\S')$
			whose value process $\int_0^\cdot K_s \vd \S'_s$
			is of finite variation and non-constant with positive probability.
			This is seen by inspecting the proof of \cite[Theorem~3.1]{Fontana15}
			(cf. \cite[Remark~3.1]{Fontana15}).
			
			It is worth noting that this equivalence is no longer true if the asset price process is allowed to have jumps.
			Indeed, let $N=(N_t)_{t\in[0,T]}$ be a Poisson process with intensity~$1$.
			Consider the compensated Poisson process model
			\begin{equation}\label{eq:311025a4}
				\S'_t=N_t-t,\quad t\in[0,T].
			\end{equation}
			As $\S'$ is a martingale, NIP holds, that is, increasing profits do not exist in this model.
			On the other hand, the strategy $H\equiv1$ produces a non-constant value process of finite variation: $V^H=\S'$.\footnote{We also observe that this strategy is \emph{admissible} in the sense that its value process is bounded from below by a deterministic constant ($T$).
				In other words, we cannot save that equivalence for c\`adl\`ag semimartingales by considering only admissibile strategies.}
		\end{remark} 
		
		In the realm of the previous remark, the following lemma reveals some intrinsic structure underlying value processes of increasing or finite variation profits.
		
		\begin{lemma} \label{lem: no domination} 
			Let \(S' = (S'_t)_{t \in [0,T]}\) be a continuous semimartingale and take \(H \in L (S')\) such that a.s. \(\int_0^\cdot H_s \vd S'_s\) is of finite variation and dominated by the quadratic variation measure~\(\rd \langle S' \rangle\). Then, a.s. \(\int_0^\cdot H_s \vd S'_s = 0\).
		\end{lemma}
		\begin{proof}
			As \(\int_0^\cdot H_s \vd S'_s\) is of finite variation, we get that a.s. 
			\[
			0 = \Big \langle \int_0^\cdot \frac{\1_{\{H_s \neq 0 \}}}{H_s} \vd \Big( \int_0^s H_r \vd S'_r \Big) \Big \rangle = \int_0^\cdot \1_{\{H_s \neq 0 \}} \vd \langle S' \rangle_s.
			\]
			Hence, by the domination assumption, a.s.
			\[
			0 = \int_0^\cdot \1_{\{ H_s \neq 0 \}} \vd \Big( \int_0^s H_r \vd S'_r \Big) = \int_0^\cdot H_s \vd S'_s,
			\]
			which completes the proof.
		\end{proof}
		
		This lemma explains that value processes of increasing (or finite variation) profits are not dominated by the quadratic variation measure of the discounted asset price process. As we encounter below, however, it is possible that such value processes are dominated by the quadratic variation measure \(\rd \langle U \rangle\) of the space-transformed natural scale diffusion \(\U = \s (\Y)\). 
		This appears to be an interesting feature of our general diffusion setting,
		and the existence of such increasing profits is related to failure of the representation property of~$\S$ (see Section~\ref{sec: RP}).

	\section{The Structure of Increasing Profits} \label{sec: main}
	In our previous paper \cite{ACU25} we established a deterministic characterization of NIP in terms of the scale function and the speed measure. 
	The proofs in \cite{ACU25} relied on the fundamental theorem of asset pricing for NIP, which states that NIP is equivalent to a weak structure condition; cf. \cite{Fontana15}. 
	In this paper we investigate the NIP condition from a quite different point of view. Namely, instead of studying the NIP condition directly, we focus on the structure of increasing profits. This path provides new economic insights and it also leads to a new proof for results from \cite{ACU25}.

	By Standing Assumption~\ref{SA: semi + boundary} and Lemma~\ref{lem: isf DC}, restricted to the open set \(\s (J^\circ)\), the inverse scale function \(\q = \s^{-1}\) is the difference of two convex functions. Consequently, the second derivative measure \(\q'' (\rd x)\) is well-defined on \(\s (J^\circ)\). 
	By Lebesgue's decomposition and the Radon--Nikodym theorem (\cite[Theorems~5.2.6, 5.3.5]{benedetto}), there exists a unique decomposition
	\[
	\q'' (\rd x)= \fqpp (x) \vd x + \q''_\si (\rd x)\quad
	\text{on }\mathcal B(\s(J^\circ)),
	\]
	where \(\q''_\si\) is a signed measure that is singular w.r.t. the Lebesgue measure \(\llambda\). For $\llambda$-a.a. $x\in\s(J^\circ)$,
	the second derivative $\q''(x)$ of $\q$ at the point $x$ exists in the usual sense,
	is finite, and $\q''(x)=\fqpp(x)$.
	Therefore, in what follows, we prefer to write $\q''(x)$ instead of $\fqpp(x)$.
	
	\begin{remark}\label{rem:251125a1}
		As $\q$ is a difference of two convex functions on $\s(J^\circ)$,
		its right and left derivatives $\q'_+$ and $\q'_-$ exist, are finite everywhere on $\s(J^\circ)$ and can differ only on an at most countable set.
		In what follows, we use the notation $\{\q'_+=0\}$
		as a shorthand for $\{x\in\s(J^\circ):\q'_+(x)=0\}$.
		Furthermore,
		we write $\{\q'=0\}$ for an arbitrary Borel subset of $\s(J^\circ)$
		that differs from the set $\{\q'_+=0\}$ on a Lebesgue-null set
		(i.e., $\llambda(\{\q'=0\} \, \triangle \,\{\q'_+=0\})=0$).
	\end{remark}
	
	Recalling \cite[Exercise~VII.3.18]{RevYor}, the process \(\U \triangleq \s (Y)\) is a general diffusion on natural scale (i.e., up to increasing affine transformations, the scale function is the identity) with speed measure \(\m^{\U} \triangleq \m \circ \s^{-1}\). 
	We denote the Lebesgue decomposition (w.r.t. the Lebesgue measure) of the
	restriction $\m^\U |_{\s(J^\circ)}$ to $(\s(J^\circ), \mathcal B(\s(J^\circ)))$ of the speed measure \(\m^{\U}\)
	by 
	\[
	\m^{\U} |_{\s (J^\circ)} (\rd x) = \fmu (x) \vd x + \m^{\U}_\si(\rd x)\quad
	\text{on } \mathcal{B}(\s(J^\circ)).
	\]
	Furthermore, we introduce the auxiliary
	signed measure $\nu$ on $(\s(J), \mathcal B (\s(J)))$ by the formula
	\begin{equation} \label{eq:150425a4}\begin{split} 
			\nu (\rd x) \triangleq - &\, \1_{\{\q'= 0\} \, \cap \, \s (J^\circ)} (x)  r \q ( x) \fmu (x) \vd x
			\\&+ \1_{\s (J^\circ)} (x) \big[ \tfrac{1}{2} \q''_\si (\rd x) - r \q ( x) \m^{U}_{\si} (\rd x)  \big]
			\\&- \1_{\{\alpha \, \in\, \mathcal{A}\}} r\alpha \, \delta_{\s (\alpha)} (\rd x)
			- \1_{\{\beta \, \in\, \mathcal{A}\}} r\beta \, \delta_{\s (\beta)} (\rd x)
			\\
			&+ \1_{\{\l \, \in \, \mathcal{R}\}} (\tfrac{1}{2}\q'_+ (\s (\l)) - r \l \m^{U} (\{\s(\l)\})) \, \delta_{\s (\l)} (\rd x) 
			\\
			&+ \1_{\{\r \, \in \, \mathcal{R}\}} (- \tfrac{1}{2} \q'_- (\s (\r)) - r \r \m^{U} (\{ \s(\r)\}))\, \delta_{\s (\r)} (\rd x),
	\end{split} \end{equation} 
	where  \(\mathcal{A}\) (resp., \(\mathcal{R}\)) denotes the set of absorbing (resp., reflecting) boundaries for the diffusion \(\Y\). Notice that \(\nu\) is locally finite on $(\s(J^\circ), \mathcal B (\s(J^\circ)))$.
	Every term of $\nu $ captures a specific effect which results in an increasing profit.
	More specifically, we will prove that NIP holds if and only if the measure $\nu $ vanishes $ (\nu \equiv 0)$. In Section~\ref{sec: examples} below, we will link each component of \(\nu\) to path properties of \(\Y\), and use examples to demonstrate how these lead to increasing profits.
	
	\smallskip 
	Our first result provides a description for the value process \(V^H\)
	associated to an increasing profit \(H \in L (\S)\).
	More generally, the result only requires a value process of finite variation. In this context, we recall that increasing profits have a natural relation to trading strategies whose value processes are of finite variation, see Remark~\ref{rem: FV profit}.
	We also define the hitting times 
	\[
	T_{x} (\U) \triangleq \inf \{t \in [0, T] \colon \U_t = x\} \wedge T, \quad \inf \emptyset \triangleq \infty,
	\]
	for \(x \in \s(J)\).
	
	Finally, for the general diffusion $\U$ on natural scale,
	we introduce the \emph{diffusion local time} as the random field
	$ (\widehat L_t^x (\U):(t,x)\in[0,T]\times\s(J))$, defined by the formula
	\begin{equation}\label{eq:251125a1}
		\widehat L_t^x (\U) =
		\begin{cases}
			L_t^x (\U)&\text{if }(t,x)\in[0,T]\times\s(J\setminus\beta),
			\\
			L_t^{x-} (\U)&\text{if }(t,x)\in[0,T]\times\{\s(\beta)\}
			\text{ (in case }
			\beta\in J
			\text{)},
		\end{cases}
	\end{equation}
	where $\beta\triangleq\sup J$ and $( L_t^x (\U) \colon (t,x)\in[0,T]\times\IR)$ is the
	semimartingale local time.
	In other words, in the diffusion local time for a general diffusion on natural scale we correct the value of the semimartingale local time only at the upper boundary of the state space
	(notice that we have $L^{\s(\beta)}(\U)=0$,
	once $\beta\in J$,
	for the semimartingale local time, as the latter is right-continuous in the space variable).
	
	\begin{proposition}\label{lem:150425a1}
		If the strategy \(H \in L (\S)\) is such that its value process $V^H$ is of finite variation, then a.s.
		\begin{equation}
			\label{eq_abs_continuity}
			V^H = \int_0^\cdot H_s e^{- rs} \int \Big( 
			\1_{\s (J \, \setminus \, \cA)} (x) \,
			\rd \widehat L^x_s (U)
			+ \1_{\s(\cA)} (x)  \1_{(T_{x} (\U), T]} (s) \vd s \Big) \, \nu (\rd x).
		\end{equation}
	\end{proposition}

	Next, we ask about the precise structure of increasing profits and, as a byproduct, a characterization of the NIP condition. This question is answered in the following theorem, which we consider as our main result. 
	We need some additional notation for its formulation.
	Let 
	\[
	\nu = \nu_+ - \nu_- 
	\] 
	be the Jordan decomposition (\cite[Theorem~5.1.8]{benedetto}) of \(\nu\) on \(\mathcal{B} (\s (J))\) and, as always, we denote the total variation measure by \(| \nu | \triangleq \nu_+ + \nu_-\). Further, let \(\s(J) = N_+ \sqcup N_-\) be a Hahn decomposition (\cite[Theorem~5.1.9]{benedetto}) for \((\s (J), \mathcal{B} (\s(J)), \nu)\), i.e., for all \( A \in \mathcal{B} (\s (J)) \), 
	\begin{align} \label{eq: Hahn decomposition}
		\nu (A \cap N_+) = \nu_+ (A), \quad \nu ( A \cap N_-) = - \nu_- (A).
	\end{align}
	Let \(\nu |_{\s (J^\circ)} = \nu_{\textup{ac}} + \nu_{\textup{si}}\) be the Lebesgue decomposition (\cite[Theorem~5.2.6]{benedetto}) of the locally finite signed measure \(\nu|_{\s (J^\circ)}\) w.r.t. the Lebesgue measure on \((\s (J^\circ), \mathcal{B} (\s (J^\circ)))\), and let \(N_{\on{si}} \in \mathcal{B} (\s (J^\circ))\) be a Lebesgue-null set such that \(\nu (A \cap N_{\on{si}}) = \nu_\textup{si} (A)\) for all \(A \in \mathcal{B}( \s (J^\circ) )\), which exists by the definition of the singular part (\cite[Definition~5.2.1]{benedetto}).
	We set
	$$
	N_{\q' = 0} \triangleq \s(J\setminus J^\circ) \cup N_{\on{si}}\cup \{ \q'_+ = 0 \} \in \mathcal{B} (\s (J)), 
	$$
	and notice that \(\nu\) is concentrated on \(N_{\q' = 0}\), i.e., \(\nu (A \cap N_{\q' = 0}) = \nu (A)\) for all \(A \in \mathcal{B} (\s (J))\). 
	Next, we define the strategy  
	\begin{align} \label{eq: theta}
		\theta_t \triangleq \1_{N_+ \, \cap \, N_{\q' = 0}} (U_t) - \1_{N_- \, \cap \, N_{\q' = 0}} (U_t), \quad t \in [0, T],
	\end{align}
	and the kernel 
	\[
	\mu (\rd t, \omega) \triangleq \int_{\s (J \, \setminus \, \cA)}
	\vd \widehat L^x_t (U)
	(\omega) \, |\nu| (\rd x)
	+ \sum_{x \in \s(\cA)} |\nu (x)| \1_{(T_{x} ( \U(\omega)), T]} (t) \vd t,
	\]
	where we again use the diffusion local time from~\eqref{eq:251125a1}.
	We write \(\mu \otimes \P\) for the measure \(\mu (\rd t, \omega) \, \P (\rd \omega)\) and \(\rd \langle \U \rangle \otimes \P\) for the measure \(\rd \langle \U \rangle_t (\omega) \, \P (\rd \omega)\), both defined on the measurable space \(([0, T] \times \Omega, \mathcal{B} ([0, T]) \otimes \cF)\).
	Finally, write 
	\[
	T_{\s (\cA)} (\U) \triangleq \inf \{t \in [0, T] \colon \U_t \in \s ( \cA ) \} \wedge T,
	\]
	for the first time \(\U\) hits an absorbing boundary point. 
	
	We now present the main results of this paper. The first theorem highlights the structural importance of the strategy \(\theta\) and the signed measure \(\nu\), providing an equivalent characterization for the NIP condition.

	\begin{theorem}\label{thm_main}
		The following are equivalent:
		\begin{enuroman}
			\item There exists an increasing profit. 
			\item There exists a set \(G \in \mathcal{B} (\s (J))\) such that \(| \nu | (G) > 0\).
			\item The trading strategy \(\theta\) from \eqref{eq: theta} is an increasing profit. 
		\end{enuroman}
		In particular, NIP is equivalent to \(| \nu | \equiv 0\). 
	\end{theorem} 
	
	It is worth noting that \(| \nu | \equiv 0\) is equivalent to \(\nu \equiv 0\).
	Indeed, if $|\nu|=\nu_++\nu_-\equiv0$, then $\nu_+\equiv0$ and $\nu_-\equiv0$, hence $\nu=\nu_+-\nu_-\equiv0$.
	Conversely, if $\nu\equiv0$, then, by~\eqref{eq: Hahn decomposition}, $\nu_+=\nu( \, \cdot \,\cap N_+)\equiv0$ and $\nu_-=-\nu(\, \cdot\,\cap N_-)\equiv0$, hence $|\nu|=\nu_++\nu_-\equiv0$.
	
	The next theorem provides a precise characterization for the set \(\IP\) of increasing profits. At this point, we recall that the corresponding value processes are described in Proposition~\ref{lem:150425a1} above.  
	
	\begin{theorem} \label{theo: characterization}
		A strategy \(H \in L (\S)\) is an increasing profit if and only if it satisfies the following three properties:
		\begin{enuroman}
			\item \(\rd \langle \U \rangle \otimes \P\)-a.e. \(H \q'_+ (\U) = 0\) on \([0, T_{\s (\cA)} (\U))\).
			\item \(\mu \otimes \P\)-a.e. \(\theta H \geq 0\). 
			\item \(\P (  \mu (\{t \in [0, T] \colon \theta_t H_t > 0 \}, \cdot \,) > 0 ) > 0\).
		\end{enuroman}
	\end{theorem}
	Providing some intuition, (i) deactivates the martingale part of the value process \(V^H\), (ii) entails that it has increasing paths, and (ii) and (iii) together ensure that \(V^H\) has a positive terminal value with positive probability.
	We remark that $\q'_+$ can be replaced with $\q'_-$ in~(i). This follows from the semimartingale occupation times formula together with the fact that $\q'_+$ and $\q'_-$ differ on an at most Lebesgue-null set (in fact, even on an at most countable set).
	
	In concrete model situations, (i) and~(ii) from Theorem~\ref{theo: characterization} are quite easy to understand, while
	condition~(iii) appears to require a case by case study.
	Our final main result replaces this part with a pathwise condition, providing a sufficient condition for a strategy to be an increasing profit.

\newpage 	
	\begin{proposition} \label{prop: suff}
		Let \(H \in L (\S)\) satisfy the following:
		\begin{enuroman}
			\item \(\rd \langle \U \rangle \otimes \P\)-a.e. \(H \q'_+ (\U) = 0\) on \([0, T_{\s (\cA)} (\U))\).
			\item \(\mu \otimes \P\)-a.e. \(\theta H \geq 0\). 
			\item There exists a set \(G \in \mathcal{B} (\s (J))\) such that \(|\nu| (G) > 0\) and a.s. it holds:
			for all \(t \in [0, T]\) with \(\U_t \in G\), we have \(\theta_t H_t > 0\).
		\end{enuroman}
		Then, \(H\) is an increasing profit.
	\end{proposition}

	\begin{discussion}
		(a) In our previous result \cite[Theorem~3.1]{ACU25}, we proved that NIP is equivalent to the following three conditions: 
		\begin{enumerate}
			\item[\textup{(i)}]
			Every accessible boundary point \(b \in J \setminus J^\circ\) satisfies one of the following two conditions: 
			\begin{enumerate}
				\item[\textup{(i.a)}] \(b\) is absorbing and either \(r = 0\) or \(b = 0\);
				\item[\textup{(i.b)}] \(b\) is reflecting and 
				\[
				r b\, \m^\U (\{ \s (b) \}) = \begin{cases} \frac12 \q'_+ (\s (\l)), & b = \l, \\[1mm] - \frac12 \q'_- (\s (\r)), & b = \r. \end{cases}
				\]
			\end{enumerate}
			\item[\textup{(ii)}]
			\(r \q (x) \m^{\U}_\si (\rd x) = \frac{1}{2} \q''_\si (\rd x)\) on \(\mathcal{B}(\s (J^\circ) )\).
			\item[\textup{(iii)}]
			\(r \fmu (x) = 0\) for $\llambda$-a.a. \(x \in \{z \in \s (J^\circ) \colon \q' (z) = 0\}\).
		\end{enumerate}
		Taking \eqref{eq:150425a4} into account, (i) holds if and only if the last three terms in \eqref{eq:150425a4} are zero, (ii) holds if and only if the second term is zero, and finally, (iii) holds if and only if the first term vanishes.
		Notice that the first term in~\eqref{eq:150425a4} contains the additional factor $\q (x)$, but it vanishes in at most one point, as $\q$ is strictly increasing, and is, therefore, excluded from~(iii).
		As a consequence, since \(| \nu | \equiv 0\) if and only if \(\nu \equiv 0\), the equivalence of (i) and (ii) from Theorem~\ref{thm_main} recovers \cite[Theorem~3.1]{ACU25}. 
		
		\smallskip
		(b) Theorem \ref{thm_main} reveals the fundamental importance of the trading strategy \(\theta\). First, there exists an increasing profit if and only if \(\theta \in \IP\) and second, any increasing profit is only made on the support of~\(\theta\), i.e., on the set \(\{ t \in [0, T] \colon \theta_t \neq 0 \}\). More specifically, if \(H\) is an increasing profit, we must have \(\P \otimes \mu\)-a.e. \(\{ H > 0 \} \subset \{ \theta \geq 0 \}\) and \(\{ H < 0 \} \subset \{ \theta \leq 0 \}\). This provides a rather precise understanding of how an increasing profit can be achieved.  Proposition~\ref{prop: suff} provides a useful recipe to design increasing profits using \(\theta\) and non-negligible sets under the measure~\(| \nu |\). 
		Recalling Remark~\ref{rem: unbounded}, scaling the strategy \(\theta\) allows to gain an unbounded increasing profit. 
		
		\smallskip
		(c) The structure of \(\theta\) and \(\nu\) connects the existence of increasing profits to certain path properties of our general diffusion \(\Y\) and further explains how they can be converted into increasing profits. We discuss these interpretations
		in Section~\ref{sec: examples} by considering a variety of examples.  
		
		\smallskip 
		(d) The Hahn decomposition \(N_+ \sqcup N_-\) is not unique (it is only unique up to null sets; see the first remark on p. 224 in \cite{benedetto} for details). As a consequence, the trading strategy \(\theta\) depends on this decomposition. Nevertheless, by virtue of Proposition~\ref{lem:150425a1}, the value function \(V^\theta\) is independent of the choice of the Hahn decomposition in the sense that all ``versions'' of \(\theta\) lead to the same value process.
		
		Another natural question is whether the choice of the set \(N_{\q' = 0}\) is unique. In general, we only require the following two properties: first, \(N_{\q' = 0} \cap \{ \q' \neq 0 \}\) must be a Lebesgue-null set and second, the signed measure \(\nu\) must be concentrated on \(N_{\q' = 0}\). The purpose of the 
		first property is to guarantee that the value process \(V^\theta\) is of finite variation, and the purpose of the second property is explained by the structure of \(V^\theta\) as in Proposition~\ref{lem:150425a1}, activating the possibility of a positive terminal value. Our choice of \(N_{\q' = 0}\) clearly has these two properties. Again, the value process is not affected by taking a different set with such properties.
		
		\smallskip 
		(e) Proposition~\ref{prop: suff} provides sufficient but not necessary conditions for an increasing profit.
		The point for including it is its simplicity in comparison with part (iii) from Theorem~\ref{theo: characterization}.
		To see that the conditions in Proposition~\ref{prop: suff} are not necessary, consider the case \(\s (J) = [ \s (\l), \infty)\) with \(\l \in \mathcal{R}\), and assume that \(G \in \mathcal{B} (\s(J))\) is such that \(|\nu| (G) > 0\). Then, for
		\(R\in(\s(x_0),\infty)\)
		large enough, we also have \(|\nu| (G \cap [\s (\l), R ]) > 0\) and, following arguments from the proof of Proposition~\ref{prop: suff}, one can show that 
		\[
		H_t \triangleq \theta_t \1_{G \, \cap \, [\s (\l), R]} (\U_t) \1_{\{ t \leq T_{R+1} (\U) \}}, \quad t \in [0, T], 
		\] 
		is an increasing profit. However, \(H_t = 0\) on \((T_{R+1} (\U), T]\), while, for any non-empty set \(A \in \mathcal{B} (\s (J))\), with positive probability, \((T_{R+1} (\U), T] \cap \{ t \in [0, T] \colon U_t \in A \} \neq \emptyset\). Again, for details we refer to the proof of Proposition~\ref{prop: suff} below. Of course, it is possible to sharpen the sufficient conditions in Proposition~\ref{prop: suff} to cover such examples.  However, it seems that a precise description is difficult to formulate without using the kernel \(\mu\) (cf. Theorem~\ref{theo: characterization}~(iii)).
		
		Last, let us stress that if the NIP condition fails, there exists always an increasing profit \(H \in \IP\) that satisfies the properties from Proposition~\ref{prop: suff}, namely the canonical strategy \(\theta\).
	\end{discussion}

	\section{Proofs of Theorems \ref{thm_main} and~\ref{theo: characterization}, and Propositions \ref{lem:150425a1} and~\ref{prop: suff}}\label{sec:proofs}
	This section is dedicated to the proofs of our main results. 
	We start with Proposition~\ref{lem:150425a1}, followed by Theorem~\ref{theo: characterization} and Proposition~\ref{prop: suff}, and finally, we turn to Theorem~\ref{thm_main}. Throughout this section, to ease our presentation, we only consider the situation \(\s (J) = [\s(\l), \infty)\).
	In this case, diffusion 
	and semimartingale local times coincide,
	and we therefore use only the latter in the formulas below.
	All other cases for $\s(J)$ can be treated similarly. 
	
	To provide fairly self-contained proofs, let us recall some results from the references~\cite{ACU25,BrugRuf16, Kal21, RW2}.
	\begin{lemma}[\textup{\cite[Lemma~3.2]{ACU25}}] \label{lem: ACU}
		Let \(\S = \S_0 + M + A\) be the Doob--Meyer decomposition of \(\S\), where \(M\) is the local martingale and \(A\) the finite variation part. 
		
		\smallskip
		\noindent
		\textup{(a)} 
		In case \(\l\) is absorbing for \(Y\), then 
		\begin{equation} \label{eq: final DMD}
			\begin{split}
				\rd \langle M\rangle_t &= e^{- 2rt} \big[\q'_+ (\U_t) \big]^2 \1_{\{t < T_{\s (\l)} (\U) \}} \vd \langle \U\rangle_t, \\ 
				\rd A_t &= e^{-rt}   \Big[ - r \q (\U_t) \vd t + \frac{1}{2} \int_{\s (J^\circ)} \rd L^x_t (\U) \, \q'' (\rd x) \Big], 
			\end{split}
		\end{equation}
		where the indicator \(\1_{\{t < T_{\s (\l)} (\U)\}}\) is included
		to emphasize that we do not require \(\q'_+ (\s(\l))\) to be well-defined (and, indeed, the limit of $\q'_+(u)$, as $u\searrow\s(\l)$, can fail to exist).
		
		\smallskip
		\noindent
		\textup{(b)}
		In case \(\l\) is reflecting for \(Y\), then 
		\begin{equation} \label{eq: final DMD - re}
			\begin{split}
				\rd \langle M\rangle_t &= e^{- 2rt} \big[\q'_+ (\U_t) \big]^2 \vd \langle \U\rangle_t, \\ 
				\rd A_t &= e^{-rt} \Big[ - r \q (\U_t) \vd t+ \frac{1}{2} \q'_+ (\s (\l)) \vd L^{\s (\l)}_t (\U) + \frac{1}{2} \int_{\s (J^\circ)} \rd L^x_t (\U) \, \q'' (\rd x) \Big].
			\end{split}
		\end{equation}
	\end{lemma} 
	
	\begin{lemma}[\textup{\cite[Lemma~3.5]{ACU25}}]\label{lem:270125a1}
		Consider an open interval $I\subset\bR$ and a function $f\colon I\to\bR$ such that
		\begin{enuroman}
			\item
			$f'$ exists and is finite \(\llambda\)-a.e. on $I$,
			
			\item
			$f''$ exists and is finite $\llambda$-a.e. on $I$.
		\end{enuroman}
		Then, $f''\1_{\{f'=0\}}=0$ $\llambda$-a.e. on $I$.
	\end{lemma}
	
	\begin{lemma}[\textup{\cite[Theorem~1.1]{BrugRuf16}}] \label{lem: BR}
		For every \(\varepsilon \in (0, T]\), \(x_0 \in \s
		(J^\circ\cup\mathcal R)\) and \(y_0 \in \s (J)\),
		\[
		\P (T_{y_0} (\U) < \varepsilon \mid \U_0 = x_0) > 0,
		\]
		where we recall that \(\mathcal{R}\) are the reflecting boundaries of \(\Y\) and \(T_{y_0} (\U) = \inf \{t \in [0, T] \colon \U_t = y_0 \}\wedge T\). 
	\end{lemma} 
	
	\begin{lemma}[\textup{\cite[Corollary~29.18]{Kal21}}] \label{lem: kallenberg}
		Let \(M = (M_t)_{t \geq 0}\) be a continuous local martingale with local time process \((L_t^x (M) \colon x \in \bR, \, t \geq 0 )\). Then, a.s., it holds simultaneously for all $x\in\bR$ and $t\in\bR_+$ that
		\[
		\{ L_t^x (M) > 0 \} = \Big\{ \inf_{s \in [0, t]} M_s < x < \sup_{s \in [0, t]} M_s \Big\}.
		\]
	\end{lemma} 
	
	The last lemma we recall is a mild extension of the diffusion occupation time formula from \cite[Theorem~V.49.1]{RW2} as provided by \cite[Lemma~C.15]{CU22}.
	
	\begin{lemma}[\textup{\cite[Lemma~C.15]{CU22}}] \label{lem: DOTF}
		a.s.\ we have
		\begin{equation}\label{eq:020423a5}
			\int_0^t f (\U_s) \vd s = \int_{\s (J)} f (y) L^y_t (\U)\, \m^U (\rd y)
		\end{equation}
		simultaneously for all \(t \in [0, T]\) and all Borel functions
		\(f \colon \s (J) \to [0,\infty]\)
		with \(f (\s (\l)) = 0\) if~\(\l \in \cA\).
	\end{lemma} 
	
	We now present the proofs for our main results. 
	
	\begin{proof}[Proof of Proposition~\ref{lem:150425a1}]
		Let us start with a general observation. Denote the Doob--Meyer decomposition of \(\S\) by \(\S = \S_0 + M + A\), where \(M\) is the local martingale and \(A\) is the finite variation part. Then, as \(V^H\) is of finite variation, 
		\[V^H - \int_0^\cdot H_s \vd A_s = \int^\cdot_0 H_s \vd M_s\] is a continuous local martingale of finite variation and hence, constant. 
		Consequently,
		\[
		V^H = \int_0^\cdot H_s \vd A_s.
		\]
		Now, we distinguish the cases where \(\l\) is absorbing or reflecting for \(\Y\). 
		
		\smallskip 
		\noindent
		{\em Case 1: \(\l\) is absorbing}. By Lemma~\ref{lem: ACU},
		\begin{align*}
			\rd A_t = e^{-rt}   \Big[ - r \q (U_t) \vd t + \frac{1}{2} \int_{\s (J^\circ)} \rd L^x_t (U) \, \q'' (\rd x) \Big].
		\end{align*}
		Using the occupation time formula for diffusions as given by Lemma~\ref{lem: DOTF}, we get that
		\[
		\q (\U_t) \1_{ \s (J^\circ)} (\U_t) \vd t = \int_{\s (J^\circ)} \q (x) \rd L^x_t (\U) \, \m^U (\rd x). 
		\]
		Hence, 
		\begin{align}
			V^H &= \int_0^\cdot H_t e^{- rt} \Big[ \1_{\{ U_t = \s (\l)\}} ( - r \l) \vd t + \int_{\s (J^\circ)} \vd L^x_t (\U) \, \Big(- r \q (x) \, \m^U (\rd x) + \tfrac{1}{2} \, \q'' (\rd x) \Big) \Big] 
			\notag
			\\
			\label{eq: first formula absorbing} 
			&= \int_0^\cdot H_t e^{- rt}  \, \Big[ \1_{[T_{\s (\l)}(\U), T]} (t) (- r \l) \vd t + \int_{\s (J^\circ)} \vd L^x_t (\U) \, \Big(- r \q (x) \, \m^U (\rd x) + \tfrac{1}{2} \, \q'' (\rd x) \Big) \Big], 
		\end{align}
		where we use that \(\l \in \cA\). As explained above, \(\int_0^\cdot H_s \vd M_s = 0\) and, in particular, \(H^2_t \vd \langle M \rangle_t = 0\). 
		Using the formula for \(\langle M \rangle\) from Lemma~\ref{lem: ACU}, and the occupation time formula for semimartingales, we get, for \(t < T_{\s (\l)} (\U)\),  
		\begin{equation} \label{eq: q' neq 0 irrelevant}
			\begin{split} 
				0 &= \int_0^t \frac{e^{2rt} \1_{\{\q' (\U_s) \neq 0\}}}{\big[\q' (\U_s) \big]^2} H^2_s \vd \langle M \rangle_s 
				\\&= \int_0^t H^2_s \frac{ \big[\q'_+ (\U_s) \big]^2 \1_{\{\q' (\U_s) \neq 0\}}}{\big[ \q' (\U_s) \big]^2 }  \vd \langle U \rangle_s 
				\\&= \int_0^t H^2_s \, \int_{\s (J^\circ)} \frac{\big[\q'_+ (x) \big]^2 \1_{\{\q' (x) \neq 0\}}}{\big[ \q' (x) \big]^2}  \vd L^x_s (\U) \vd x 
				\\&= \int_0^t H^2_s \, \int_{\s (J^\circ)} \1_{\{\q' (x) \neq 0 \}} \vd L^x_s (\U) \vd x. 
			\end{split}
		\end{equation} 
		Using this identity and the fact that \(\llambda\)-a.e. \(\1_{\{\q' = 0\}} \q'' = 0\) by Lemma~\ref{lem:270125a1}, \eqref{eq: first formula absorbing} reformulates to \eqref{eq_abs_continuity} and the formula is proved. 
		
		\smallskip
		\noindent
		{\em Case 2: \(\l\) is reflecting}. Again by Lemma~\ref{lem: ACU}, 
		\begin{align*}
			\rd A_t &= e^{-rt} \Big[ - r \q (U_t) \vd t+ \frac{1}{2} \q'_+ (\s (\l)) \vd L^{\s (\l)}_t (U) + \frac{1}{2} \int_{\s (J^\circ)} \rd L^x_t (U) \, \q'' (\rd x) \Big].
		\end{align*}
		As \(\alpha \in \mathcal{R}\), the occupation time formula for diffusions given by Lemma~\ref{lem: DOTF} yields that 
		\[
		\q (\U_t) \vd t = \int_{\s (J)}  \q (x) \vd L^x_t (\U) \, \m^U (\rd x), 
		\]
		and consequently, 
		\[
		\rd A_t = e^{-rt} \Big[ \big( \tfrac{1}{2} \q'_+ (\s (\l)) - r \l \m^U (\{ \s (\l) \}) \big)\vd L^{\s (\l)}_t (U) +  \int_{\s (J^\circ)} \rd L^x_t (U) \, \big( \tfrac{1}{2}\q'' (\rd x) - r \q (x) \, \m^U (\rd x) \big)\Big].
		\] 
		Finally, as \eqref{eq: q' neq 0 irrelevant} holds irrespective of the boundary classification of \(\l\), and \(\llambda\)-a.e. \(\1_{\{\q' = 0\}} \q'' = 0\), again by Lemma~\ref{lem:270125a1}, the formula \eqref{eq_abs_continuity} follows.
	\end{proof} 
	
	\begin{proof}[Proof of Theorem~\ref{theo: characterization}]
		First, we prove the necessity of the conditions (i)-(iii), assuming \(H \in \IP\). As the value process \(V^H\) of the increasing profit \(H\) is of finite variation, Lemma~\ref{lem: ACU} yields that  
		\[
		0 = \langle V^H \rangle_T = \int_0^T H^2_s \vd \langle \S \rangle_s = \int_0^T e^{- 2rs} \big( H_s \q'_+ (\U_s) \big)^2 \1_{\{s < T_{\s (\cA)} (\U)\}} \vd \langle \U\rangle_s. 
		\]
		Notice the use of $T_{\s(\cA)}(\U)$ rather than $T_{\s(\l)}(\U)$ in the above display. This comprises both cases of Lemma~\ref{lem: ACU} into one formula.
		The latter display implies that \(\rd \langle \U\rangle \otimes \P\)-a.e. \(H \q'_+ (\U) = 0\) on \([0, T_{\s (\cA)} (\U))\), i.e., part (i) holds. 
		We next establish part~(ii). 
		As \(V^H\) is of finite variation, we can apply Proposition~\ref{lem:150425a1}, which yields that
		\begin{align*}
			V^H_t = \int_0^t H_s e^{- rs} \int \Big( \1_{\s (J \, \setminus \, \cA)} (x) \, \rd L^x_s (U) + \1_{\s(\cA)} (x)  \1_{(T_{x} (\U), T]} (s) \vd s \Big) \, \nu (\rd x).
		\end{align*} 
		Using that \(\nu\) is concentrated on \(N_{\q' = 0}\), the identity \(\theta^2 = \1_{N_{\q' = 0}} (\U)\), the identities in \eqref{eq: Hahn decomposition}, and the occupation time formula for semimartingales, we compute that, for \(t \in [0, T]\), 
		\begin{equation} \label{eq: computation to mu integral} 
			\begin{split}
				V^H_t
				&= 	\int_0^t \theta_s H_s e^{- rs} \int \Big( \1_{\s (J \, \setminus \, \cA)} (x) \, \rd L^x_s (U) + \1_{\s(\cA)} (x)  \1_{(T_{x} (\U), T]} (s) \vd s \Big) \, |\nu| (\rd x)
				\\&= 	\int_0^t \theta_s H_s e^{- rs} \, \mu (\rd s, \cdot \,).
			\end{split}
		\end{equation} 
		As \(V^H\) is an increasing process, the same holds for
		\[
		\int_0^\cdot e^{rs} \vd V^H_s = \int_0^\cdot \theta_s H_s \, \mu(\rd s, \cdot \,).
		\]
		The standard measure theory yields that \(\mu \otimes \P\)-a.e. \(\theta H \geq 0\), which means that (ii) holds. 
		Finally, as \(\P(V^H_T > 0) > 0\), we also have 
		\[
		\P \Big( \int_0^T \theta_s H_s e^{- rs} \mu (\rd s, \cdot \,) > 0 \Big) > 0, 
		\]
		which clearly implies
		\[
		\P ( \mu (\{ t \in [0, T] \colon \theta_t H_t > 0 \}, \cdot \,) > 0 ) > 0. 
		\] 
		In summary, (i)-(iii) hold, completing the proof of the necessity direction. 
		
		\smallskip
		We turn to the proof of the converse direction, assuming that \(H \in L (\S)\) satisfies (i)-(iii). 
		In the following we show \eqref{eq: computation to mu integral}, 
		which implies \(H \in \IP\). More precisely, then (ii) implies that \(V^H\) is increasing and (iii) implies that it has a positive terminal value with positive probability. 
		
		As explained in the first part of this proof, to get \eqref{eq: computation to mu integral} it suffices to prove that \(V^H\) is of finite variation. Using property (i) and Lemma~\ref{lem: ACU}, we get that \(\langle V^H \rangle = 0\), which implies that \(V^H\) is of finite variation. 
	\end{proof} 
	
	\begin{proof}[Proof of Proposition~\ref{prop: suff}]
		By virtue of Theorem~\ref{theo: characterization}, it suffices to understand the implication
		\begin{equation} \label{eq: implication to show}
			\begin{split}
				\exists \, G \in \mathcal{B} (\s (J)) \colon | \nu | (G) > 0&, \, \text{a.s.} \ \theta H > 0 \text{ on } \{t \in [0, T] \colon U_t \in G \}
				\\&\implies 	\P ( \mu (\{ t \in [0, T] \colon \theta_t H_t > 0 \}, \cdot \,) > 0 ) > 0.
			\end{split}
		\end{equation} 
		We assume that \(G\) is as in \eqref{eq: implication to show}.
		By the monotone convergence theorem, 
		\[
		|\nu |(G \cap [\s (\l), R]) \to |\nu | (G) > 0, \quad R \nearrow \infty.
		\]
		Hence, there exists an \(R \in (\s (x_0), \infty)\) such that \(| \nu | (G \cap [\s (\l), R]) > 0\). 
		We now prove that 
		\begin{align} \label{eq: measure positive}
			\P (\mu ( \{ t \in [0, T] \colon U_t \in G \cap [\s (\alpha), R] \}, \cdot \,) > 0) > 0, 
		\end{align} 
		which implies the implication~\eqref{eq: implication to show} and hence, completes the proof. 
		Our argument is split into three steps. Before we start, recall that 
		\begin{align*}
			\mu (\{ t \in [0, T]\colon  & \U_t \in G \cap [\s (\l), R]\}, \cdot \,) 
			\\&= \int_{G \, \cap \, ([\s (\l), R] \setminus \s (\cA))} L^x_T (\U) \, |\nu| (\rd x) + |\nu (G \cap \s (\cA))| \big( (T_{\s (\l)} (\U)\vee T) - T_{\s(\l)} (\U)\big).
		\end{align*}  
		\smallskip 
		\noindent
		{\em Step 1:}
		If \(|\nu | (G \cap \s (\cA)) > 0\), then statement~\eqref{eq: measure positive} follows from Lemma~\ref{lem: BR},
		as the latter implies \(\P (T_{\s (\l)} (\U) < T) > 0\).
		
		In the following two steps, we will show that 
		\begin{align} \label{eq: positive local time}
			\P ( L^x_T(\U) > 0 \text{ for all } x \in  [\s (\l), R] \setminus \s (\cA)  ) > 0,
		\end{align} 
		which entails \eqref{eq: measure positive}
		whenever \(| \nu |(G \cap \s (\cA)) = 0\), as then \(|\nu| \big( G \cap ([ \s (\l), R] \setminus \s (\cA))  \big) > 0\).
		
		\smallskip
		\noindent
		{\em Step 2:} We now show that a.s. on \(\{T_{\s (\l)}(\U) < T,\, T_{R} (\U) < T\}\), 
		\[
		L^x_T (\U) > 0 \text{ for all \(x \in [\s (\l), R] \setminus \s(\cA)\)}.
		\] 
		To see this, recall from \cite[Theorem~33.9]{Kal21} that there exists a Brownian motion \(B = (B_s)_{s \geq 0}\) (possibly on an extended probability space) such that a.s. \(\U_t = B_{\gamma_t}\) for \(t \in [0, T]\), where
		\[
		\gamma_t \triangleq \inf \Big\{ s \geq 0 \colon \int_{\s(J)} L^{x}_s (B) \, \m^\U (\rd x) > t \Big\}. 
		\]
		Moreover, it is easy to see that \(t \mapsto \gamma_t\) is a.s. strictly increasing on the set \([0, T_{\s (\cA)} (\U))\). We now distinguish the cases where \(\l\) is reflecting or absorbing. 
		
		{\em Case 1: \(\l \in \mathcal{R}\)}. On the set \(\{T_{\s (\l)}(\U) < T,\, T_{R} (\U) < T\}\), we have a.s. \(\gamma_T > T_{\s (\l)} (B) \vee T_{R} (B)\) and hence, by standard properties of Brownian paths, a.s. 
		\[
		\min_{s \in [ 0, \gamma_T ]} B_s < \s (\l), \quad \max_{s \in [0, \gamma_T]} B_s > R.
		\] 
		Consequently, Lemma~\ref{lem: kallenberg} yields that, a.s. on \(\{T_{\s (\l)}(\U) < T,\, T_{R} (\U) < T\}\), \(L_T^x (\U) = L_{\gamma_T}^x (B) > 0\) for all \(x \in [\s (\l), R]\). 
		
		{\em Case 2: \(\l \in \cA\)}. While we still have a.s. \(\gamma_T > T_R (B)\) on \(\{T > T_R (\U)\}\), the difference to the previous case is that only a.s. \(\gamma_T = T_{\s (\l)} (B)\) on \(\{T > T_{\s (\l)} (\U) \}\). Thus, we can only conclude that, a.s. on  \(\{T_{\s (\l)}(\U) < T,\, T_{R} (\U) < T\}\),
		\[
		\min_{s \in [ 0, \gamma_T ]} B_s = \s (\l), \quad \max_{s \in [0, \gamma_T]} B_s > R.
		\] 
		Using Lemma~\ref{lem: kallenberg}, we still get, a.s. on \(\{T_{\s (\l)}(\U) < T,\, T_{R} (\U) < T\}\), \(L_T^x (\U) > 0\) for all \(x \in (\s (\l), R]\), which is precisely what we claimed.
		
		\smallskip
		\noindent
		{\em Step 3:} We now prove that \(\P (T_{\s (\l)} (\U) < T, \, T_R (\U) < T ) > 0\). 
		Using the strong Markov property of \(\U\) and Lemma~\ref{lem: BR}, we get that
		\begin{equation*} \begin{split} 
				\P (T_{\s (\l)} (\U) < T,\, T_{R} (\U) < T) 
				&\geq \P (T_{\s (\l)} (\U_{\cdot + T_{R} (\U)}) < T / 2, \, T_{R} (\U) < T/2 ) 
				\\&= \P (T_{\s (\l)} (\U) < T / 2 \mid \U_0 = R) \P ( T_{R} (\U) < T / 2) > 0.
		\end{split} \end{equation*}
		In summary, \eqref{eq: positive local time} follows and the proof complete. 
	\end{proof} 
	
	\begin{proof}[Proof of Theorem~\ref{thm_main}]
		The implication (iii) \(\implies\) (i) is trivial, and (i) \(\implies\) (ii) follows by contraposition: if \(| \nu | \equiv 0\), then \(\nu \equiv 0\) and Proposition~\ref{lem:150425a1} yields that \(\IP = \emptyset\). 
		
		It remains to prove the implication (ii) \(\implies\) (iii). 
		We use Proposition~\ref{prop: suff}, verifying \ref{prop: suff}-(i)-(iii) for \(H = \theta\).
		Using the occupation time formula for semimartingales, we obtain that 
		\begin{align*}
			\int_0^T | \theta_s \q'_+ (\U_s) | \vd \langle \U \rangle_s = \int \1_{N_{\q' = 0}} (x) |\q'_+ (x)| L^x_T (\U) \vd x = \int \1_{\{\q'_+ (x) = 0\}} |\q'_+ (x)| L^x_T (\U) \vd x = 0, 
		\end{align*}
		which proves \ref{prop: suff}-(i).
		Next, \ref{prop: suff}-(ii) holds trivially, since \(\theta^2 \geq 0\).
		Finally, take \(G \in \mathcal{B} (\s (J))\) with $|\nu| (G) > 0$.
		Then, \(| \nu | (G \cap N_{\q' = 0}) > 0\), because \(| \nu |\) is concentrated on \(N_{\q' = 0}\), and since \(\theta^2 = \1_{N_{\q' = 0}} (\U) > 0\) on \(\{ t \in [0, T] \colon \U_t \in G \cap N_{\q' = 0} \}\), \ref{prop: suff}-(iii) holds. The proof is complete. 
	\end{proof}

	\section{Examples} \label{sec: examples}
	The structure of \(\theta\) and \(\nu\) connects the existence of increasing profits to path properties of our general diffusion \(\Y\), explaining how they can generate strategies that are increasing profits. 
	In the following, we illustrate these connections through a variety of examples.
	For reader's convenience, we recall that $\nu$ is the signed measure on $(\s(J),\mathcal B(\s(J)))$ that is given by the formula
	\begin{align*}
		\nu (\rd x) = - & \, \1_{\{\q'= 0\} \, \cap \, \s (J^\circ)} (x) r \q ( x) \fmu (x) \vd x \\
		&+ \1_{\s (J^\circ)} (x) \big[ \tfrac{1}{2} \q''_\si (\rd x) - r \q ( x) \m^{U}_{\si} (\rd x)  \big]
		\\&- \1_{\{\alpha \, \in\, \mathcal{A}\}} r\alpha \, \delta_{\s (\alpha)} (\rd x)
		- \1_{\{\beta \, \in\, \mathcal{A}\}} r\beta \, \delta_{\s (\beta)} (\rd x)
		\\
		&+ \1_{\{\l \, \in \, \mathcal{R}\}} (\tfrac{1}{2}\q'_+ (\s (\l)) - r \l \m^{U} (\{\s(\l)\})) \, \delta_{\s (\l)} (\rd x) 
		\\
		&+ \1_{\{\r \, \in \, \mathcal{R}\}} (- \tfrac{1}{2} \q'_- (\s (\r)) - r \r \m^{U} (\{ \s(\r)\}))\, \delta_{\s (\r)} (\rd x). 
	\end{align*}
	
	We start discussing the situation where \(\Y\) attains boundary points, which possibly activates some of the last three lines in the definition of \(\nu\). The first example discusses cases with absorbing boundaries,
	which turns out to be the case with the easiest interpretation.
	
	\begin{example}[Engelbert--Schmidt diffusion market model]
		Consider a classical SDE model
		\[
		\rd \Y_t = b (\Y_t) \vd t + \sigma (\Y_t) \vd W_t, \quad \Y_0 = y_0 \in J^\circ \triangleq (\l, \r), 
		\] 
		where \(W = (W_t)_{t \in [0, T]}\) is a Brownian motion and the coefficients \(b \colon J^\circ \to \bR\) and \(\sigma \colon J^\circ \to \bR\) are Borel functions that satisfy the Engelbert--Schmidt conditions
		$$
		\forall \, x \in J^\circ: \;
		\sigma ( x ) \neq 0, \quad
		\frac{1+|b|}{\sigma^2} \in L^1_\textup{loc} (J^\circ). 
		$$
		We assume that \((\Y, W)\) is a weak solution to the above SDE and we stipulate that \(\Y\) gets absorbed in \(\{\l, \r\}\) when it hits this set.
		This convention is in conjunction with the classical Engelbert--Schmidt theory \cite{EngSch}; see Chapter~5.5 in \cite{KarShr} for a textbook presentation.
		The sets \(\cA\) and \(J\triangleq J^\circ \cup \cA\) are determined via $b$ and $\sigma$ by Feller's test for explosion; see \cite[Proposition~5.5.29]{KarShr}.
		Moreover, we need to choose $b$ and $\sigma$ in such a way that we get $J\subset\bR$
		(notice that the latter is necessary for Standing Assumption~\ref{SA: semi + boundary}),
		i.e.,
		explosions at infinite boundary points are not allowed.
		To be more specific, with 
		\[
		\s (x) \triangleq \int^x \exp \Big\{ - \int^y \frac{2 b (z)}{\sigma^2 (z)} \vd z \Big\} \vd y \text{ for }  x \in (\l, \r), \text{ and } \s (\l) \triangleq \lim_{y \searrow \l} \s (y), \ \, \s (\r) \triangleq \lim_{y \nearrow \r} \s (y), 
		\]
		for every \(b \in \{ \l, \r \} \setminus \bR\) one of the following two items must hold: 
		\begin{enumerate} 
			\item[(a)] \(| \s (b) | = \infty\);
			\item[(b)] $|\s(b)|<\infty$ and, for every non-empty interval \(I \subset J^\circ\) with \(b\) as endpoint, \[\int_I \frac{ | \s (b) - \s (y) | }{\s' (y) \sigma^2 (y)} \vd y = \infty. \]
		\end{enumerate}
		This guarantees that the solution process \(Y\) cannot reach the values \(\pm \infty\). In the same spirit, we characterize
		\[
		\cA = \big\{ b \in \{ \l, \r \} \cap \bR \colon \text{both (a) and (b) above fail}\, \big\}.  
		\]
		It is well-known (\cite{EngSch}) that \(\Y\) is a general diffusion with scale function \(\s\) and speed measure 
		\[
		\m (\rd x) \triangleq \frac{\rd x}{\s' (x) \sigma^2 (x)}\text{ on } \mathcal{B} (J^\circ), \qquad \m (\{ b \} ) \triangleq \infty
		\;\;\forall \, b \in \cA. 
		\]
		Lastly, let us comment on Standing Assumption~\ref{SA: semi + boundary}, i.e., when $\Y$ is a semimartingale.
		By definition of a weak solution up to explosion (see Definition~5.5.20 in~\cite{KarShr}), the solution process \(Y\) is always a semimartingale on the stochastic interval \([0, T_{\cA} (Y))\). However, it is a delicate point that the semimartingale property can get lost at the hitting time of a boundary point,
		see the counterexamples in \cite[Section~4]{MU2015}. 
		We give a flavor for conditions on \(b\) and \(\sigma\) that entail Standing Assumption~\ref{SA: semi + boundary}: if \(\l \in J\) but \(\r \not \in J\), then it holds if and only if 
		\[
		\int_{\l + } \frac{| \s (y) - \s (\l) |}{\s' (y)} \frac{| b (y) |}{\sigma^2 (y)} \vd y < \infty
		\]
		(cf. \cite{CinJacProSha80} or \cite[Corollary~3.6]{MU2015}).

		Next, it is straightforward to prove that the function $\q=\s^{-1}$ is continuously differentiable with absolutely continuous derivative.
		In particular, the measures $\q''(\rd x)$ and $\m^\U=\m\circ\s^{-1}$ are absolutely continuous w.r.t. the Lebesgue measure.
		Moreover, as, for all \(x \in \s (J^\circ)\),
		\[ 
		\q' (x) = \exp \Big\{ \int^{\q (x)} \frac{2 b (z)}{\sigma^2 (z)} \vd z \Big\} > 0,
		\] 
		we obtain that
		\[
		\nu (\rd x) = - \sum_{b\in\cA} rb \, \delta_{\s (b)} (\rd x). 
		\] 
		As
		a consequence of Theorem~\ref{thm_main}, NIP holds if and only if \(rb = 0\) for all
		\(b \in \cA = J \setminus J^\circ\). This recovers the well-known fact that NIP holds in the zero interest rate regime, but we also observe that there are increasing profits in the non-zero interest rate regime when \(\Y\) has non-zero (necessarily absorbing) boundary points. 
		For suitable choices of \(N_+, N_-\), depending on the sign of \(rb \neq 0\),
		$b\in\cA$,
		and
		\(N_{\q' = 0} = \s(\cA)\), we find that 
		\[
		\theta = - \sum_{b \in \cA} \, \on{sgn} (rb) \1_{\{ \U\, =\, \s (b) \} } = - \sum_{b \in \cA} \, \on{sgn} (rb) \1_{[T_{\s (b)} (\U), T]}, 
		\]
		is an increasing profit with value process 
		\[
		V^\theta_t = \sum_{b \in \cA} \int_{T_{\s (b)} (\U) }^{T_{\s (b)} (\U) \vee\, t} |br| e^{- rs} \vd s
		= \sum_{b \in \cA} \big| b ( e^{- r T_{\s (b)} (\U) } - e^{- r (T_{\s (b)} (\U) \vee \, t)}) \big|.
		\]
		To
		get the idea behind these results, assume that \(rb \neq 0\) for some $b\in\cA$.
		In this case, on the time interval \([T_{\s (b)} (\U), T]\), the discounted price process \(\S_t = e^{- rt} \Y_t\) is non-constant and either increasing or decreasing. Thus, we achieve an increasing profit either by buying or selling the risky asset at time \(T_{\s (b)} (\U)\). This is exactly what \(\theta\) suggests (and its sign accounts for buying or selling). 
	\end{example}
	
	We now turn to the case with reflecting boundary points, which deals with the final two terms of the auxiliary signed measure \(\nu\). We notice that these have two different ingredients, namely \(\q'_{\pm} (\s (b))\) and \(r b \m^U (\{ \s (b) \})\), where \(b \in \mathcal{R}\). Both of these terms turn out to be related to local time effects, but with a quite different flavor. The first one represents local time terms in the drift that account for reflection from the boundary, while the second term measures stickiness in the reflecting boundary. 
	
	\begin{example}[Black--Scholes model with reflection] \label{ex: BS reflection}
		Consider a version of the Black--Scholes model with reflection from a positive boundary. This model has been studied in \cite{BDH23}, where it was shown that NIP fails for this model in the zero interest rate regime. 
		We introduce the model through its scale function and speed measure, taking
		\(J = [1, \infty)\), $\mu\in\bR$, $\sigma>0$ and 
		\begin{align*}
			\s (x) &= \int^x y^{- 2 \mu / \sigma^2} \vd y, \ x \in [1, \infty),
			\\
			\m (\rd x) &=
			\frac{x^{2 \mu / \sigma^2 - 2}}{\sigma^2} \vd x
			\text{ on } \mathcal{B} ((1, \infty)), \quad \m (\{ 1 \}) \in [0, \infty).
		\end{align*} 
		The value \(\m (\{ 1 \}) \in [0, \infty)\) decides about the reflective behavior of the model. If \(\m (\{ 1\} ) = 0\) the reflection is instantaneous, as in the paper \cite{BDH23}, and if \(\m (\{ 1 \} ) > 0\), the boundary point \(1\) is sticky reflecting. 
		Again, the second derivative measure \(\q''(\rd x)\) is absolutely continuous w.r.t. the Lebesgue measure, \(\q' > 0\) with
		\(\q'_+ (\s(1)) = 1\),
		and we get that 
		\begin{align*}
			\nu (\rd x) = \Big(\frac{1}{2} - r \m (\{ 1 \}) \Big)\, \delta_{\s (1)} (\rd x),
		\end{align*}
		as, clearly, $\m^\U(\{\s(1)\})=\m(\{1\})$.
		Using Theorem~\ref{thm_main}, we obtain that 
		\[
		\text{NIP holds} \quad \Longleftrightarrow \quad r \m ( \{ 1 \} ) = \frac{1}{2}.
		\]
		In particular, if \(r = 0\) or \(\m (\{ 1 \}) = 0\), there exists an increasing profit, covering the observation from \cite{BDH23}.
		If NIP fails, 
		\[
		\theta =
		\on{sgn} \Big(\frac{1}{2} - r \m (\{ 1 \}) \Big)
		\1_{\{\Y \, = \, 1\}}
		\] 
		is an increasing profit with value process 
		\[
		V^\theta_t = \Big| \frac{1}{2} - r \m (\{ 1 \}) \Big| \, \int_0^t e^{- r s} \vd L^{\s (1)}_s (\U) = \Big| \frac{1}{2} - r \m (\{ 1 \}) \Big| \, \int_0^t e^{- rs} \vd L^1_s (\Y), \quad t \in [0, T]. 
		\]
		Notice that $L^1(\Y)=L^{\q(\s(1))}(\q(\U))=L^{\s(1)}(\U)$
		due to $\q'_+(\s(1))=1$ and \cite[Exercise VI.1.23]{RevYor} together with Lemma~\ref{lem: isf DC}
		(the latter justifies the application of \cite[Exercise VI.1.23]{RevYor}). The structure of \(\theta\) explains that increasing profit can only be made on the set \(\{ t \in [0, T] \colon \Y_t = 1\}\). Namely, depending on the sign of \(\frac{1}{2}  - r \m (\{ 1 \})\), buying or selling while \(\Y\) is in its reflecting state~\(1\) yields an increasing profit. This observation is also reflected by the fact that, whenever \(H\) is an arbitrary increasing profit, the value process of \(H\) is given by
		\[
		V^H_t = \int_0^t H_s e^{- rs} \Big( \frac{1}{2} - r \m (\{ 1 \}) \Big) \vd L^1_s (\Y), \quad t \in [0, T]
		\]
		(cf. Proposition~\ref{lem:150425a1}).
		Let us discuss more intuitions behind this observation.
		As will be shown in a forthcoming paper \cite{CUZ_25+},
		the dynamics of \(\Y\) can be described via an SDE with constraints, namely 
		\begin{align}\label{eq: dynamics sticky reflection}
			\rd \Y_t = \Y_t \1_{\{\Y_t \neq 1\}} \Big( \mu \vd t + \sigma \vd W_t \Big) + \frac{1}{2} \vd L^1_t (\Y), \quad \1_{\{\Y_t = 1\}} \vd t = \m (\{ 1 \}) \vd L^1_t (\Y).
		\end{align} 
		By the integration by parts formula and the side constraint, 
		\begin{align*}
			\rd \S_t
			&=e^{- rt} \Y_t \1_{\{\Y_t \neq 1\}} \Big( \mu \vd t + \sigma \vd W_t \Big)
			+ \frac{1}{2} \, e^{- rt} \vd L^1_t (\Y) - r e^{- rt} \Y_t \vd t
			\\&=e^{- rt} \Y_t \1_{\{\Y_t \neq 1\}} \Big( (\mu - r) \vd t + \sigma \vd W_t \Big)
			+ e^{- rt} \, \Big( \frac{1}{2} - r \m (\{ 1 \}) \Big) \vd L^1_t (\Y).
		\end{align*} 
		The first term provides another reasoning for our observation that any increasing profit must be supported on \(\{t \in [0, T]\colon \Y_t = 1\}\), as otherwise the martingale part gets activated, while the second term explains the condition \(r \m (\{ 1 \}) \neq 1 / 2\), as otherwise \(\S\) is constant on the set \(\{ t \in [0, T] \colon \Y_t = 1\}\).
		Notice also that in case \(r \m ( \{ 1 \} ) = 1 / 2\) the local time effects from skewness and stickiness cancel each other, and that then no increasing profit exists.
	\end{example} 
	
	In the previous example, there was a boundary point \(b\) with \(\q'_{\pm} (\s (b)) \neq 0\), leading to a local time term in the drift.
	The following example illustrates that this term might also be inactive, initiating increasing profits solely through the stickiness at the reflecting boundary. 
	
	\begin{example}[Shifted generalized square Bessel process of dimension \(\delta \in (0, 2)\)]\label{ex: SSQB}
		We consider a shifted generalization of the square Bessel process of low dimension \(\delta \in (0, 2)\) that features sticky reflection, while the classical square Bessel process only allows for instantaneous reflection. 
		We define \(\Y\) with state space \(J = [1, \infty)\) through scale and speed given by 
		\begin{align*}
			\s (x) &= (x - 1)^{1 - \delta / 2}, \ x \in [1, \infty), \\
			\m (\rd x) &=
			\frac{(x - 1)^{\delta/2 - 1}}{4 (1-\frac\delta2)} \vd x
			\text{ on } \mathcal{B}((1, \infty)), \quad \m (\{ 1 \} ) \in [0, \infty).
		\end{align*}
		In this case, the inverse scale function is given by 
		\begin{align*}
			\q (x) = x^{1 / (1 - \delta / 2)} + 1, \quad x \in \s ([1, \infty)) = \bR_+.
		\end{align*} 
		Using $1/(1-\delta/2)>1$, it follows that \(\q'_+ (\s (1)) = \q'_+ (0) = 0\), and the auxiliary signed measure \(\nu\) is given by the formula
		\[
		\nu (\rd x) = - r \m (\{ 1 \} ) \, \delta_{0}(\rd x).
		\]
		As a consequence,
		$$
		\text{NIP holds}
		\quad\Longleftrightarrow\quad
		r \m (\{ 1 \}) = 0.
		$$
		In the case \(r \m (\{ 1 \}) \neq 0\), we may take 
		\[
		\theta =
		- \on{sgn} (r \m (\{ 1 \}) ) \1_{\{ \Y \, = \, 1 \}},
		\]
		explaining that investing while \(\Y\) is in the reflecting boundary leads to an increasing profit. 
		To provide some heuristic intuition, in contrast to the previous example, \(\q'_+ (0) = 0\) deactivates the local time term in the dynamics that accounts for reflection (relating this to the previous example, this corresponds to the term \(\rd L^1_t (\Y) / 2\) in \eqref{eq: dynamics sticky reflection}). Still, a local time term arises through stickiness (as in the previous example, this corresponds to the second equation in \eqref{eq: dynamics sticky reflection}).\footnote{To stress even more the difference from the previous example, we notice that, in the present example, there is no increasing profit in the instantaneously reflecting case $\m(\{1\})=0$.}
		In the end, we have
		\[
		\1_{\{ \Y_t \, = \, 1\}} \vd \S_t = e^{- rt} (- r \m (\{ 1 \})) \vd L^{0}_t (\U),  
		\]
		which is either increasing or decreasing, and non-constant if \(r \m (\{ 1 \}) \neq 0\).
	\end{example}
	
	Next, we discuss the influence of the second term from \(\nu\).
	Recall that $N_{\on{si}}\in\mathcal B(\s(J^\circ))$ denotes a Lebesgue-null set such that
	$\nu(A\cap N_{\on{si}})=\nu_{\on{si}}(A)$ for all $A\in\mathcal B(\s(J^\circ))$.
	First consider the case where $N_{\on{si}}$ consists of one point, say $a\in\s(J^\circ)$.
	Again, the singular parts \(\q''_{\on{si}}(\rd x)\) and \(\m^\U_{\on{si}}(\rd x)\) account for local time terms that occur in the drift of $\S$.
	In the case of \(\m^\U_{\on{si}}(\rd x)\) these come from sticky behavior of $\Y$ at the point $\q(a)$, precisely as this was the case with reflecting boundaries. For \(\q''_{\on{si}}(\rd x)\) the situation is different, and the local time terms arise from the desire of the process $\Y$ to leave the point $\q(a)$ in a preferred direction.
	The following two examples illustrate these phenomena.
	
	\begin{example}[Bachelier model with stickiness] \label{ex: sticky Bach}
		We consider an extension of the classical Bachelier model,
		where the price of the risky asset is a Brownian motion with a sticky point $\xi \in \bR$, i.e., the general diffusion on $\IR $ with scale and speed defined by 
		\begin{equation}
			\s(x) = x, \quad 
			\m(\rd x) = \vd x + \rho \, \delta_{\xi}(\rd x),
		\end{equation}
		where $\rho \ge 0 $ is the so-called stickiness parameter.
		To give some intuition on this process, the amount of time the process spends at $\xi $ is governed by the relation
		\begin{equation}
			\int_{0}^{t}\indic{\{ \Y_s = \xi \}} \vd s = \rho L^{\xi}_t(\Y), \quad \text{for all } t\ge 0,
		\end{equation}
		and, in the case \(\rho> 0\), is of positive Lebesgue measure, as soon as the threshold $\xi $ is reached. 
		
		Then, the auxiliary measure is given by
		$$
		\nu(\rd x) = - r\xi \rho\, \delta_{\xi}(\rd x),
		$$ 
		and Theorem~\ref{thm_main} shows that
		$$
		\text{NIP holds}
		\quad\Longleftrightarrow\quad
		r\xi\rho=0.
		$$
		In the case $r\xi\rho\ne0$, the canonical increasing profit is given by
		\(\theta = - \on{sgn} (r \xi) \1_{\{ Y_t \, =\, \xi \}}\),
		and its value process reads
		\[
		V^\theta_t = \int_0^t |r \xi| \rho e^{- rs} \vd L^{\xi}_s (\Y),
		\quad t\in[0,T].
		\]
		Similarly, by Proposition~\ref{lem:150425a1}, whenever \(H\) is an increasing profit, its value process is given by
		\[
		V^H_t = - \int_0^t H_s r \xi \rho e^{- rs} \vd L^{\xi}_s (\Y),
		\quad t\in[0,T].
		\]
		We notice that increasing profits can only be made in the sticky point.
		To give some intuition on this example, we recall that $\Y $ solves the SDE system involving the local time (see \cite{EngPes}):
		\begin{equation}
				\vd \Y_t = \indic{\{ \Y_t \not = \xi \}} \vd B_t, \quad
				\indic{\{ \Y_t = \xi \}} \vd t = \rho \vd L^{\xi}_t (\Y), 
		\end{equation}
		where $B$ is a Brownian motion.
		Integration by parts yields that the discounted price process $\S $ follows the dynamic
		\begin{align}
			\vd \S_t &= \vd (e^{-rt} \Y_t) = \indic{\{ \Y_t \not = \xi \}} \left( -r \S_t \vd t + e^{-rt} \vd B_t \right)
			- r \xi \rho e^{-rt} \vd L^{\xi}_t (\Y).
		\end{align}
		Therefore, $\S $ has a local time drift which, for the same reason as in Examples~\ref{ex: BS reflection} and \ref{ex: SSQB},
		initiates an increasing profit.
	\end{example}
	
	\begin{example}[Bachelier model with skewness]
		\label{ex: skew}
		We take a look at the Bachelier model with skewness, as considered in \cite{Rossello2012}. In other words, we suppose that \(\Y\) is a Brownian motion with skewness at zero, i.e., the state space is given by \(J = \bR\) and scale and speed are defined through
		\begin{align*}
			\s ( x ) &= \begin{cases} (1 - \kappa) x, & x \geq 0, \\ \kappa x,& x < 0, \end{cases}
			\\
			\m (\rd x) &=
			(1 - \kappa)^{-1} \1_{\{x \geq 0\}} \vd x + \kappa^{-1} \1_{\{ x < 0\}} \vd x,
		\end{align*}
		where \(\kappa \in (0, 1)\setminus\{1/2\}\) is the so-called skewness parameter. To provide some intuition, if \(\kappa < 1 / 2\) (resp., $\kappa>1/2$), then the process \(\Y\) has the tendency to leave the origin downwards (resp., upwards).
		We notice that \(\q' > 0\) $\llambda$-a.e.
		and \(\q''_{\on{si}} (\rd x) = (2 \kappa - 1)/(\kappa (1 - \kappa)) \, \delta_0 (\rd x)\). Consequently, 
		\[
		\nu (\rd x) = \frac{2 \kappa - 1}{2 \kappa (1 - \kappa)} \, \delta_0 (\rd x), 
		\]
		and Theorem~\ref{thm_main} shows that
		\(\theta = \on{sgn} (2\kappa-1) \1_{\{\Y = 0\}}\)
		is an increasing profit, independently of the interest rate. The value process of \(\theta\) is given by 
		\[
		V^\theta_t = \int_0^t \frac{|2 \kappa - 1|}{2\kappa (1 - \kappa)} e^{- rs} \vd L^{0}_s (\U)
		= \int_0^t \Big| 1 - \frac{1}{2 \kappa} \Big| e^{- rs} \vd L^0_s (\Y),
		\quad t\in[0,T],
		\] 
		where we use \cite[Exercise VI.1.23]{RevYor} for the last identity.
		Similarly, whenever \(H\) is an increasing profit, its value process reads
		\[
		V^H_t = \int_0^t \Big( 1 - \frac{1}{2 \kappa}\Big) H_s e^{- rs} \vd L^0_s (\Y),
		\quad t\in[0,T].
		\]
		We notice that increasing profits can only be made in the skew point. 
		To provide some intuition for the origin of increasing profits, recall from \cite[Exercise~X.2.30]{RevYor} that the process \(\Y\) has an SDE representation of the form 
		\[
		\rd \Y_t = \rd W_t + \Big( 1 - \frac{1}{2 \kappa}\Big) \vd L^0_t (\Y).
		\]
		This formula explains that \(\Y\) has a local time drift, which stems from the fact that \(\q''_{\on{si}} (\{ 0 \}) \neq 0\). For the same reason as in the Examples~\ref{ex: BS reflection}, \ref{ex: SSQB} and \ref{ex: sticky Bach}, this drift initiates an increasing profit.
	\end{example} 
	
	In general, the phenomena that stem from the second term from \(\nu\) are richer than the one described in the previous example.
	But, in any case, the corresponding increasing profits are made on the set
	$\{t\in[0,T] \colon \U_t\in N_{\on{si}}\}$
	(cf.~\eqref{eq: theta}),
	where we also recall that $N_{\on{si}}\subset\s(J^\circ)$ is
	Lebesgue null.\footnote{To provide a specific example where the sets of the form $\{t\in[0,T] \colon \U_t=a\}$ with $a\in\s(J^\circ)$ do not suffice,
		consider $r\ne0$ and a general diffusion $\Y$ with $J=\bR$ on natural scale (in particular, $\Y=\U$)
		with the speed measure that does not have atoms but has a nonvanishing singular component $\m^{\U}_{\on{si}}(\rd x)$ concentrated on the Cantor set.}
	This is different from the influence of the first term from \(\nu\), which appears to be the most curious one and is discussed in the next example.
	
	\begin{example}[Increasing profits made on a set of positive Lebesgue measure] \label{ex: positive lebesgue measure} 
		We start by constructing a scale function on \(J = \bR\), following an idea from \cite[Lemma~2.1]{CU24}.
		Take a closed set $F\subset[0,1]$ with empty interior such that $\llambda(F)>0$.
		This could be a \emph{fat} Cantor set or, alternatively, one could construct such a set as follows (cf. \cite[Example~1.7.6]{bogachev}).
		Let $\{q_n\colon n\in\mathbb N\}$ be an enumeration of all rational points in $[0,1]$.
		Take $a\in(0,1)$ and a positive sequence $\{r_n\colon n\in\mathbb N\}$ such that
		$\sum_{n=1}^\infty 2r_n\le a$.
		It is easy to verify that
		$F\triangleq[0,1]\setminus G$,
		where $G\triangleq\bigcup_{n\in\mathbb N}(q_n-r_n,q_n+r_n)$,
		satisfies the requirements.
		Now, we set 
		$$
		\q (x) \triangleq \int_0^x d_F (z) \vd z,
		\quad x \in \bR,
		\quad d_F (z) \triangleq \inf_{y \in F} |z - y|.
		$$
		Notice that $\q$ is a $C^1$-function on $\bR$ with
		\begin{equation}\label{eq:070225a1}
			\{x\in\bR\colon\q'(x)=0\}=F
		\end{equation}
		(because $z\mapsto d_F (z)$ is continuous and $F$ is closed) and
		$\q$ is strictly increasing (because $F$ is closed and does not contain any open interval).
		Let \(\U\) be a Brownian motion and define \(\Y \triangleq \q (\U)\), 
		which is a general diffusion
		with state space $\q(\bR)=\bR$, scale function $\s\triangleq\q^{-1}$ and speed measure $\llambda\circ\q^{-1}$.
		Furthermore, as $\q'=d_F(\cdot)$ is Lipschitz continuous on $\bR$,
		hence absolutely continuous (in particular, $\q''_{\on{si}}(\rd x)\equiv0$),
		then $\q'$ is of locally finite variation, showing that $\q$ is a difference of two convex functions on $\bR$.
		This explains that Standing Assumption~\ref{SA: semi + boundary} is satisfied.
		As $\m^\U=\llambda$, we also have $\m^\U_{\on{si}}(\rd x)\equiv0$, hence
		\[
		\nu (\rd x) = - r \q (x) \1_F (x) \vd x. 
		\] 
		By Theorem~\ref{thm_main},
		$$
		\text{NIP holds}
		\quad\Longleftrightarrow\quad
		r = 0.
		$$
		In the case $r\ne0$, we may take
		\(
		\theta =  - \on{sgn} (r) \1_{F} (\U),
		\)
		which is an increasing profit with the value process
		\begin{align*}
			V^\theta_t
			&=
			\int_0^t |r| e^{- rs} \int_F \q (x) \vd L^x_t (\U) \vd x
			\\
			&=
			\int_0^t |r| e^{- rs} \q(\U_s) \1_F (\U_s) \vd \langle \U \rangle_s
			\\
			&=
			\int_0^t |r| e^{- rs} \Y_s \1_F (\U_s) \vd \langle \U \rangle_s,\quad t\in[0,T].
		\end{align*}
		In contrast to the previous examples, increasing profits are not made while \(\Y\) attains a discrete set of points, but while \(\U\) is in the uncountable set \(F\) with positive Lebesgue measure. Another crucial difference is that in this setting the value processes of increasing profits are not integrals w.r.t. a local time process,
		but w.r.t. the quadratic variation process $\langle\U\rangle$.
	\end{example}

	\section{On the Relation of Increasing Profits and the Representation Property}\label{sec: RP}
	
	Example~\ref{ex: positive lebesgue measure} illustrated a very peculiar form of increasing profits whose value processes are integrals w.r.t. the quadratic variation measure \(\rd \langle \U\rangle\). As we will encounter in this section, for a variety of general diffusion models, the existence of such increasing profits is intrinsically connected to the failure of the so-called representation property (RP),
	which is of fundamental importance both in the context of market completeness (see, e.g., \cite[Section VII.2.d]{shir}) and from the viewpoint of the general theory of stochastic processes (see, e.g., \cite[Sections III.4.c-d]{JacodShiryaev2003}).
	
	To study this connection formally, we define $(\cF^{S}_{t})_{t\in [0,T]} $ to be the right-continuous natural filtration of \(\S\), that is, \(\cF^\S_t \triangleq \bigcap_{s \in (t, T]} \sigma (\S_r, r \leq s)\), for all $t\in [0,T)$, and $\cF^\S_T \triangleq \sigma (\S_r, r \leq T)$.
	Recall from Stricker's lemma (\cite[Theorem~9.19]{jacod79}) that \(\S\) is not only an \((\mathcal{F}_t)_{t \in [0, T]}\)-semimartingale, but also an \((\cF^S_t)_{t \in [0, T]}\)-semimartingale.
	
	\begin{definition}[Representation property]
		We say that the {\em representation property (RP)} holds for the semimartingale \(\S\) if every \((\cF^\S_t)_{t \in [0, T]}\)-local martingale \(M = (M_t)_{t \in [0, T]}\) has a representation 
		\[
		M = M_0 + \int_0^\cdot H_s \vd \S^c_s, 
		\]
		where \(H\) is an \((\cF^\S_t)_{t \in [0, T]}\)-predictable process such that a.s. \(\int_0^T H^2_s \vd \langle \S\rangle_s < \infty\) and \(\S^c\) is the continuous \((\cF^S_t)_{t \in [0, T]}\)-local martingale part of~\(\S\).
	\end{definition}
	
	Using the main result from \cite{CU24}, in our general diffusion framework, the RP can be described in terms of the inverse scale function \(\q\).
	We recall that $\{\q'=0\}$ denotes an arbitrary Borel subset of $\s(J^\circ)$
	that differs from the set
	$\{x\in\s(J^\circ):\q'_+(x)=0\}$
	(abbreviated as $\{\q'_+=0\}$)
	on a Lebesgue-null set;
	cf. Remark~\ref{rem:251125a1}.
	
	\begin{lemma} \label{lem: RP}
		The RP holds for the semimartingale \(\S\) if and only if \(\llambda (\q' = 0) = 0\). 
	\end{lemma} 
	\begin{proof}
		Clearly, the right-continuous natural filtration
		$(\cF^\S_t)_{t \in [0, T]}$
		of $\S$ coincides with
		the right-continuous natural filtration
		$(\cF^\Y_t)_{t \in [0, T]}$ of~$\Y$.
		Using that 
		\[
		\rd \S_t = e^{- rt} \vd \Y_t - r \S_t \vd t, 
		\] 
		we observe that the continuous \((\cF^\S_t)_{t \in [0, T]}\)-local martingale part of $\S$ is given by \(\rd \S^c_t = e^{- rt} \vd \Y^c_t\), where \(Y^c\) is the continuous \((\cF^\Y_t)_{t \in [0, T]}\)-local martingale part of~\(Y\).
		
		Consequently, the RP holds for \(\S\) if and only if it holds for \(\Y\). As the RP for \(\Y\) is equivalent to \(\llambda (\q' = 0) = 0\) by \cite[Theorem~2.1]{CU24}, the claim follows. 
	\end{proof}
	
	By virtue of formula~\eqref{eq:150425a4} for the signed measure \(\nu\),
	the condition \(\llambda (\q' = 0) > 0\), characterizing 
	the failure of the RP for $\S$,
	is closely related to the existence of certain increasing profits, like the one illustrated by Example~\ref{ex: positive lebesgue measure}.
	We now describe this relation precisely.
	
	\begin{definition}
		We call an increasing profit \(H \in \IP\) a {\em quadratic variation increasing profit (QVIP)} if a.s. \(\rd V^H \ll \rd \langle \U\rangle\).
		We denote
		\[
		\QVIP \triangleq \big\{ H \in L (\S) \colon H \text{ is a QVIP} \big\}.
		\]
	\end{definition} 
	It is instructive to relate this definition to Lemma~\ref{lem: no domination}, which entails that the value processes of increasing profits cannot be dominated by the quadratic variation measure \(\rd \langle \S \rangle\). As we have seen in Example~\ref{ex: positive lebesgue measure}, this is not the case for \(\rd \langle \U \rangle\) and QVIPs may exist.
	
	Recall that \(\nu |_{\s (J^\circ)} = \nu_{\textup{ac}} + \nu_{\textup{si}}\) denotes the Lebesgue decomposition of \(\nu|_{\s (J^\circ)}\) w.r.t. the Lebesgue measure and that \(N_{\on{si}} \in \mathcal{B} (\s (J^\circ))\) is a Lebesgue-null set such that \(\nu (A \cap N_{\on{si}}) = \nu_\textup{si} (A)\) for all \(A \in \mathcal{B}( \s (J^\circ) )\).
	
	The existence of a QVIP can be characterized in the spirit of Theorem~\ref{thm_main}.
	We define the strategy 
	\[
	\otheta \triangleq \1_{(N_+ \, \cap \, \{ \q'_+ = 0\} ) \, \setminus\, N_{\on{si}}} (\U) - \1_{(N_- \, \cap \, \{ \q'_+ = 0 \}) \,  \setminus\, N_{\on{si}}} (\U), 
	\]
	which takes over the role of \(\theta\) in Theorem~\ref{thm_main}.
	
	\begin{lemma} \label{lem: QVRP}
		The following are equivalent:
		\begin{enuroman}
			\item A QVIP exists.
			\item There exists a \(G \in \mathcal{B} (\s (J^\circ))\) such that \(|\nu_{\on{ac}}|(G)>0\).
			\item \(\otheta \in \IP\).
			\item \(\otheta \in \QVIP\). 
		\end{enuroman}
		In particular, $\QVIP=\emptyset$ is equivalent to $|\nu_{\on{ac}}|\equiv0$.
	\end{lemma}
	
	It is worth recalling that $|\nu_{\on{ac}}|\equiv0$ is equivalent to $\nu_{\on{ac}}\equiv0$ (cf. the paragraph after Theorem~\ref{thm_main}).
	
	\begin{proof}
		In the proof we use the notation $N^c_{\on{si}}\triangleq\s(J^\circ)\setminus N_{\on{si}}$.
		
		The implication (iv) \(\implies\) (i) is trivial.
		Let us show that (i) \(\implies\) (ii).
		Assume that \(H \in \IP\) satisfies \(\rd V^{H}_t\ll \rd \langle \U\rangle_t\).
		By the semimartingale occupation time formula, $\langle\U\rangle_t = \int L_t^x(\U) \vd x$, that is,
		\(\rd V^H_t \ll \int \rd L^x_t (\U) \vd x\).
		As $\s(J\setminus J^\circ)\cup N_{\on{si}}$ is a Lebesgue-null set,
		then, by Proposition~\ref{lem:150425a1},
		we obtain that
		\begin{equation*} \begin{split}
				V^{H} &= \int_0^\cdot \1_{N^c_{\on{si}}} (\U_s) \vd V^H_s 
				\\&= \int_0^\cdot H_s e^{- rs} \int_{N^c_{\on{si}}} \vd L^x_s (\U) \, \nu (\rd x) 
				\\&= \int_0^\cdot H_s e^{- rs} \int \vd L^x_s (\U) \, ( \1_{(N_+ \, \cap \, \{ \q'_+ = 0 \}) \setminus N_{\on{si}}} - \1_{(N_- \, \cap \, \{ \q'_+ = 0 \}) \setminus N_{\on{si}}} ) \, | \nu | (\rd x)
				\\&=  \int_0^\cdot H_s e^{- rs} \int \vd L^x_s (\U) \, ( \1_{(N_+ \, \cap \, \{ \q'_+ = 0 \}) \setminus N_{\on{si}}} - \1_{(N_- \, \cap \, \{ \q'_+ = 0 \}) \setminus N_{\on{si}}} ) \, | \nu_{\on{ac}} | (\rd x)
				\\&= \int_0^\cdot \otheta_s H_s e^{- r s} \int \rd L^x_t (\U)\, | \nu_{\on{ac}} | (\rd x).
			\end{split} 
		\end{equation*}
		As \(H \in \IP\) implies \(\P (V^H_T > 0) > 0\), arguing by contraposition, it follows that \(| \nu_{\on{ac}} |\) cannot be the zero measure, which is equivalent to (ii).
		
		\smallskip
		Next, the implication (ii) \(\implies\) (iii) follows directly from Proposition~\ref{prop: suff}. To be more specific, notice that
		\(\otheta \q'_+(\U) = 0\), \(\otheta \theta = \1_{\{\q'_+ = 0\} \, \setminus \, N_{\on{si}}} (U) \geq 0\) and \(> 0\) on \(\{t \in [0, T] \colon U_t \in \{q'_+ = 0\} \, \setminus \, N_{\on{si}} \}\).
		As 
		\[
		G \in \mathcal{B} (\s (J^\circ)), \, | \nu_{\on{ac}} | (G) > 0 \quad \implies \quad | \nu | (G \cap \{\q'_+ = 0\} \cap N^c_{\on{si}}) = | \nu_{\ac} | (G) > 0, 
		\] 
		we conclude that (i)-(iii) from Proposition~\ref{prop: suff} hold.
		
		\smallskip			
		Finally, we prove the implication (iii) \(\implies\) (iv). 
		If \(\otheta \in \IP\), using Proposition~\ref{lem:150425a1} and the occupation time formula for semimartingales, we get that
		\begin{equation}
			\label{eq: VP QVIP} \begin{split}
				V^{\otheta} &= \int_0^\cdot e^{- rs} \int_{\{ \q'_+ = 0 \} \, \cap \, N^c_{\on{si}}} \vd L^x_s (\U) \, (\1_{N_+} (x) - \1_{N_-} (x)) \, \nu (\rd x)
				\\&= \int_0^\cdot e^{- rs} \int_{\{ \q'_+ = 0 \} \, \cap \, N^c_{\on{si}}} \vd L^x_s (\U) \, | \nu | (\rd x)
				\\&= \int_0^\cdot e^{- rs} \int_{\{ \q'_+ = 0\}} \vd L^x_s (\U) \, | \nu_{\on{ac}} | (\rd x) 
				\\&= \int_0^{\cdot} e^{- rs} \int  | r \q (x) | \m^U_{\on{ac}} (x) \1_{\{ \q' = 0\}} \vd L^x_s (\U) \vd x
				\\&= \int_0^\cdot e^{-rs} | r \q (\U_s) | \m^U_{\on{ac}} (\U_s) \1_{\{\q'  = 0\}} (\U_s) \vd \langle \U\rangle_s, 
			\end{split}
		\end{equation}
		which is the value process of a QVIP. This means that \(\otheta\) is a QVIP, i.e., (iv) holds.	
		This concludes the proof. 
	\end{proof}

	\begin{corollary} \label{coro: QVIP; RP}
		A QVIP exists if and only if \(r \neq 0\) and \(\llambda (\m^U_{\ac} > 0, \q' = 0) > 0\).
	\end{corollary}
	
	The result is a direct consequence of Lemma~\ref{lem: QVRP} and the structure of $\nu_{\on{ac}}$
	(recall~\eqref{eq:150425a4} and notice that the factor $\q(x)$ vanishes at most in one point, hence does not matter).
	
	\begin{discussion}
		(a)
		Corollary~\ref{coro: QVIP; RP} allows us to relate the failure of the RP for $\S$ to the existence of a QVIP. Namely, if, for example, \(\m^U_{\on{ac}} > 0\) on \(\{\q' = 0\}\), then it follows from Corollary~\ref{coro: QVIP; RP} and Lemma~\ref{lem: RP} that a QVIP exists if and only if \(r \neq 0\) and the RP fails. 
		
		More precisely, Corollary~\ref{coro: QVIP; RP} and Lemma~\ref{lem: RP} show that, if \(r \neq 0\), the failure of the RP for $\S$ is a necessary condition for the existence of a QVIP. To give an example where it is not sufficient, assume that \(r \neq 0\), \(\s (J) = \bR\) and that \(\m^U\)
		is a discrete measure concentrated
		on the set of rational numbers \(\mathbb{Q}\). Then, \(\m^U_{\on{ac}} = 0\) and no QVIP exists by Corollary~\ref{coro: QVIP; RP}, irrespectively of the choice of \(\q\).
		However, taking \(\q\) as in Example~\ref{ex: positive lebesgue measure}, the RP for $\S$ fails by Lemma~\ref{lem: RP}. 
		
		\smallskip
		(b)
			Corollary~\ref{coro: QVIP; RP} can be related to Lemma~\ref{lem: no domination}. Indeed, by Lemma~\ref{lem: ACU}, a.s. \(\rd \langle \S \rangle \sim \rd \langle \U \rangle\) on \(\{ t \in [0, T] \colon \q' (\U_t) > 0 \}\). In view of Lemma~\ref{lem: no domination}, this means that a QVIP can only be made on the set \(\{ t \in [0, T] \colon \q' (U_t) = 0 \}\), which implies that \(\{ \q' = 0 \}\) has to have positive Lebesgue measure by the semimartingale occupation time formula.

		\smallskip
		(c)
		Finally, we provide some intuition behind the connection between the failure of the RP for $\S$ and the existence of a QVIP.
		To this end, we first observe that
		\begin{equation}\label{eq:241125a1}
			\text{the RP for $\S$ fails}
			\quad\Longleftrightarrow\quad
			\text{the local martingale }
			\int_0^\cdot \1_{\{\q' = 0\}} (\U_s) \vd \U_s^c
			\text{ is non-constant.}
		\end{equation}
		Indeed, by the semimartingale occupation time formula together with Lemma~\ref{lem: kallenberg}, the latter is seen to be equivalent to $\llambda(\q'=0)>0$, which is, in turn, equivalent to the failure of the RP for $\S$ by Lemma~\ref{lem: RP}.\footnote{Alternatively,
			it is instructive to deduce the implication
			($\Longleftarrow$) in~\eqref{eq:241125a1} from Lemma~\ref{lem: ACU}, which, in particular, explains that the local martingale
			$\int_0^\cdot \1_{\{\q' = 0\}} (\U_s) \vd \U_s^c$
			cannot be represented as a \(\rd \S^c\)-integral if it is non-constant.}
		The nontriviality of the local martingale in~\eqref{eq:241125a1}
		is equivalent to the fact that \(\int_0^T \1_{\{\q' = 0\}} (\U_s) \vd \langle \U\rangle_s > 0\) with positive probability.
		By virtue of~\eqref{eq: VP QVIP}, this is necessary for the existence of a QVIP
		(and even sufficient if $r\ne0$ and $\m^\U_{\on{ac}}>0$ on $\{\q'=0\}$).
	\end{discussion}

	\bibliographystyle{siam}
	\bibliography{bibfile}
	
\end{document}